\documentclass[12pt]{article} 
\usepackage{amsmath, amsfonts,amsthm}
\usepackage{graphicx, wrapfig}
\usepackage{caption}
\usepackage{subcaption}
\usepackage{setspace}
\usepackage{color}

\usepackage{tikz}
\usepackage{tikz-qtree}

\usetikzlibrary{calc}
\usetikzlibrary{positioning}
\usetikzlibrary{arrows}

 \DeclareGraphicsExtensions{.eps, .ps}

\def\Y{{\bf Y}}
\def\y{{\bf y}}

\def\P{\mathcal{P}}
\def\S{\mathcal{S}}
\def\M{{\bf M}}
\def\given{\, | \,}
\def\pr{\text{pr}}
\def\id{\text{id}}
\def\bft{{\bf t}}
\def\bfw{{\bf w}}
\def\bfu{{\bf u}}
\def\bfs{{\bf s}}
\def\Nat{\mathbb{N}}
\def\Real{\mathbb{R}}
\def\indep{\mathrel{\rlap{$\perp$}\kern1.6pt\mathord{\perp}}}

\usepackage[margin=1.5in]{geometry}
\usepackage[sectionbib]{natbib}
\usepackage{hyperref}

\newtheorem{theorem}{Theorem}
\newtheorem{lemma}{Lemma}
\newtheorem{corollary}{Corollary}

\theoremstyle{definition}
\newtheorem{definition}{Definition}
\newtheorem{example}{Example}

\begin{document}
\title{Exchangeable Markov multi-state survival processes}
\author{Walter Dempsey\thanks{Department of Statistics, Harvard
    University, One Oxford Street, Cambridge, MA 02138, U.S.A. E-mail:
wdempsey@fas.harvard.edu}}
\maketitle

\begin{abstract}
We consider {\em exchangeable Markov
multi-state survival processes} -- temporal processes 
taking values over a state-space~$\S$ with at least one absorbing
failure state~$\flat \in \S$ that satisfy natural invariance
properties of exchangeability and consistency under subsampling.
The set of processes contains many well-known examples 
from health and epidemiology -- survival,
illness-death, competing risk, and comorbidity processes;
an extension leads to recurrent event processes.

We characterize exchangeable Markov
multi-state survival processes in both discrete 
and continuous time.
Statistical considerations impose natural constraints
on the space of models appropriate for applied
work.  In particular, we describe constraints 
arising from the notion of {\em composable systems}.
We end with an application of the developed models 
to irregularly sampled and potentially censored multi-state survival data,
developing a Markov chain Monte Carlo algorithm for posterior
computation. 
\end{abstract}


\section{Introduction}\label{section:introduction}
In many clinical survival studies, a recruited patient's health status is
monitored on either a regular or intermittent schedule until either
(1) an event of interest (e.g., failure) or (2) the end of the study
window. Covariates are recorded at the time of recruitment, and
treatment protocols per patient are presumed known at baseline. In the
simplest survival study, health status $Y(t)$ at time~$t$ is a binary
variable, dead~$(0)$ or alive~$(1)$. In clinical trials with health
monitoring, $Y(t)$~is a more detailed description of the state of
health 
of the individual, containing relevant patient information, e.g.,
pulse rate, cholesterol level, cognitive score or CD4 cell count
\citep{Diggle2008, Farewell2010,Kurland2009}.

In this paper, we consider the setting in which the health process takes
values in some pre-specified ``state-space''.  For example,
in the illness-death model, we summarize the current state of the 
participant as taking one of three 
possible values~$\{ \text{Healthy, Ill, Dead} \}$.
Such a process can be thought of as a coarse view of the state of health for 
a patient over time.  
The continuing importance of multi-state processes in applications 
cannot be overstated~\citep{Jepsen2015, denHout2016}. 

When no baseline covariates are measured beyond the initial
state~$Y(0)$, the model for the set of patient state-space processes
should satisfy natural constraints.  First, the model should be
agnostic to patient labeling.  Second, the model should be agnostic to
sample size considerations.
These natural constraints (mathematically defined in
section~\ref{section:state_space_models}) lead to the concept of
\emph{exchangeable Markov multi-state survival processes}.  
The purpose of this paper is to characterize the
set of multi-state survival processes and show how this theory 
of exchangeable stochastic processes fits
naturally into the applied framework of event-history analysis.
Both the parametric continuous-time Markov process with 
independent participants and the nonparametric counting 
process are contained as limiting cases.
Next, we discuss the notion of ``composable systems'' and its effect on
model specification. A Markov Chain Monte Carlo (MCMC) algorithm is
then derived for posterior computations given irregularly sampled
multi-state survival data. We end with an application to a cardiac
allograft vasculopathy (CAV) multi-state survival study.

\subsection{Related work}

Odd Aalen was one of the first to recognize 
the importance of incorporating the ``theory of 
stochastic processes'' into an ``applied framework
of event history analysis''~\cite[p. 457]{Aalen2008}. 
Martingales and counting processes form the basis
of this nonparametric approach.  The fundamental
concept of the product-integral unifies discrete
and continuous-time processes.  
Nonparametric methods, however, do not adequately handle intermittent
observations.  For example,~\citet{Aalen2015} consider dynamic path
analysis for a liver cirrhosis dataset.
In this study, the prothrombin index, a composite blood coagulation
index related to liver function, is measured initially at three-month
intervals and subsequently at roughly twelve-month intervals.
To deal with the intermittency of observation times,~\citet{Aalen2015}
use the ``last-observation carried forward'' (LOCF) assumption. 
However, such an assumption is unsatisfactory for highly variable
health processes, and can lead to biased estimates~\citep{PNASreport}.
\citet{Jepsen2015} discuss the importance of multi-state and competing
risks models for complex models of cirrhosis progression.  Here again,
the nonparametric approach assumes observation times correspond to
transition times of the multi-state process.

One alternative is to consider parametric models such as
continuous-time Markov processes. 
Prior work~\citep{Saeedi2011, Hajiaghayi2014, RaoTeh2013} has focused on 
estimation of parametric continuous-time Markov processes under
intermittent observations.
Most parametric models, however, make strong assumptions about the
underlying state-space process; in particular, most models assume
independence among patients.
One implication is that observing sharp changes in health among prior
patient trajectories at a particular time since recruitment will not
impact the likelihood of a similar sharp change in a future patient at
the same timepoint. 
The proposed approach in this paper balances between the nonparametric
and parametric approaches.



%



\section{Multi-state survival models}\label{section:state_space_models}

In this section we formally define the multi-state survival process and the 
notions of exchangeability, Kolmogorov-consistency,
and the Markov property.  We combine these 
in section~\ref{section:framework} to provide characterization
theorems for these processes
in discrete and continuous-time.


\subsection{Multi-state survival process}\label{section:state_space_data}

Formally  the {\em multi-state survival process},~$\Y$, is a function
from the label set~$\Nat \times \mathcal{T}$
into the state space~$\mathcal{S}$.
For now, we assume the cardinality is finite (i.e., ~$|\mathcal{S}| <\infty$).
If the response is in {\em discrete-time}, 
then the process is defined on $\mathcal{T} = \Nat$.
If the response is in {\em continuous-time}
then the process is defined on $\mathcal{T} = \Real^{+}$.
Each label is a pair $(u, t)$, and the value 
$\Y(u,t)$ is an element of~$\mathcal{S}$ corresponding to the
state of patient~$u$ at time~$t$.

The distinguishing characteristic of survival processes 
is \emph{flatlining}~\citep{DempseyLDA}; that is, there
exists an absorbing state~$\flat \in \S$ such that~$Y(u, t) = \flat$
implies~$Y(u, t^\prime) = \flat$ for all~$t^\prime > t$.
Thus, the survival time~$T_u$ for unit~$u$ is a deterministic function of the
multi-state survival process~$\Y$:
\[
T_u = \inf \{ t \geq 0 : Y (u, t) = \flat \}.
\]
For all~$u \in \Nat$, we assume $Y(u, 0) \neq \flat$ at recruitment,
so $T_u > 0$.  Multiple absorbing states~$\{\flat_c\}$ representing
different terminal events may occur, such as for competing risk
processes. 

Without loss of generality, we 
assume $\mathcal{S} = \{1,\ldots, s\} =: 
[s]$.  
For example, if the state-space is~$\mathcal{S} = 
\{ \text{Alive}, \text{Dead} \}$,
we recode this to $[2] = \{ 1, 2 \}$.
In this case,~$\flat = 2$ is the flatlining state.
At each time~$t$, the population-level
process is given by~$\Y (t) = \{ Y(u,t) \given u \in \Nat\}
\in [s]^{\mathbb{N}}$. 
We write $y$ to denote a generic
element of $[s]^{\mathbb{N}}$.
We write~$\Y_{A}$ 
to denote the restriction
of the state space process
to~$u \in A \subset \mathbb{N}$.
We call~$\Y_{[n]}$ the \emph{$n$-restricted
state-space process} for~$[n] := \{1,\ldots,n\}$. We write~$y_{[n]}$
to denote a generic element of~$[s]^{n}$.



\subsection{Transition graphs for multi-state survival process}

The \emph{transition graph} represents the set of potential
transitions between elements of the multi-state survival process out
of the set of $s^2$ possible transitions.
The transition graph is a directed graph~$G = (V,E)$. The vertex
set~$V = [s]$ is all potential states; the directed edge set~$E$
contains all edges $(i,j)$ such that at jump times the Markov process
can jump from $i$ to $j$.  In the illness-death model, for example, a
patient can jump from  {\em Alive} to {\em Ill} but not back;
therefore (Alive, Ill) is in the edge set but not (Ill,
Alive). Example~\ref{example:idp} in supplementary
section~\ref{section:examples} provides additional details. In the
bi-directional illness-death 
model, both edges are present in the transition graph. In
continuous-time, jumps can only occur between distinct states, in
which case the number of possible edges is $s (s-1)$. 
In this case,~$(i,i) \not \in E$ for all~$i \in V$. Any absorbing
state~$i \in [s]$ satisfies~$(i,j) \not \in E$ for all~$j \neq i \in
[s]$. We write~$\P_{G}$ to denote the set of $|S|$ by $|S|$ transition
matrices~$P$ satisfying $\sum_{j \in V} P_{i,j} = 1$,~$P_{i,j} \geq 0$
for all $i,j \in V$, and $P_{i,j} = 0$ for all $(i,j) \not \in E$.
In the continuous-time setting, define~$P_{i,i} = 1 - \sum_{j, (i,j)
  \in E} P_{i,j}$.

\subsection{Consistency under subsampling}

We first note that sample size is often an arbitrary choice based on
power considerations and/or patient recruitment
constraints. Statistical models should be agnostic to such
considerations. That is, observing $n$ units versus $n+1$ units and
then restricting to the first~$n$ units should be equivalent, i.e.,
the model should exhibit \emph{consistency under subsampling}.

Consider the multi-state survival process~$\Y_{[m]}$ for~$m > n$. Define the
restriction operator~$R_{m,n}$ to be the restriction of~$\Y_{[m]}$ to
the first $n$ individuals. Then the process is \emph{consistent under
  subsampling} if $R_{m,n} (\Y_{[m]} )$ is a version of~$\Y_{[n]}$ for
all~$[m] \supset [n]$. Under the consistency assumption, the
process~$\Y_{[n]}$ satisfies \emph{lack of interference};
mathematically,
\[
  \pr ( \Y_{[n]} \in A \given H_{[m]}(t)) 
  = \pr ( \Y_{[n]} \in A \given H_{[n]} (t) )
\]
where $H_{[l]} (t)$ is the $\sigma$-field generated by the variables
$Y(i,s)$ for $i\in [l]$ and $s\le t$. Lack of interference is
essential, ensuring the $n$-restricted multi-state survival process is
unaffected by the multi-state survival process for subsequent components.
Consistency under subsampling ensures the statistical models
are embedded in suitable structures that permit extrapolation.
Without it, one is forced into the awkward situation of 
only being interested in the current collected dataset.
Consistency ensures statistical inference to be a special
case of ``reasoning by induction''~\citep{Finetti1972}.

\subsection{Exchangeability}

We next note that the patient labels,~$u \in \Nat$, are also
arbitrary. Therefore, any suitable multi-state survival process
must be agnostic to patient relabeling.
We define a multi-state survival process~$\Y$  
to be {\em [partially] exchangeable} 
if for any permutation
$\sigma : [n] \to [n]$, the
relabeled process ${\Y}^{\sigma}_{[n]} 
= \{ Y(\sigma(1), t), \ldots, Y(\sigma(n), t) \given t \in \mathcal{T} \}$ 
is a version of~${\Y}_{[n]}$.

\subsection{Time-homogeneous Markov process}

We say~$\Y_{[n]}$ is a {\em time-homogeneous Markov process} if, for
every $t, t^\prime \geq 0$, the conditional distribution of~$\Y_{[n]}
(t+t^\prime)$ given the multi-state survival process history up to
time~$t$,~$\mathcal{H}_{[n]} (t)$, only depends on $\Y_{[n]} (t)$ and
$t^\prime$. This Markovian assumption is a simplifying assumption
which leads to mathematically tractable conclusions. In this paper, we
restrict our attention to time-homogeneous processes; therefore, we
simply say $\Y_{[n]}$ is Markovian.


\section{Markov, exchangeable multi-state survival processes}\label{section:framework}

We define a multi-state survival process that is Markovian,
exchangeable, and consistent under sampling as a {\em Markov,
  exchangeable multi-state survival process}. 
Below, we characterize these processes in both discrete and continuous
time. The behavior is markedly different in each setting with
continuous-time Markov processes exhibiting much more complex behavior
-- allowing for both single-unit changes at time~$t$ as well as
positive fraction changes -- showing why choice of time-scale matters in
applied settings.

\subsection{Discrete-time multi-state survival models}

In discrete-time, the Markov,
exchangeable multi-state survival process is governed by a series of
random transition matrices~$P_t$ each drawn independently from a
probability measure~$\Sigma$ on $\P_G$.
That is, the initial state~$\Y (0)$ is
drawn from an exchangeable distribution on~$[s]$.
Then at time~$t$, the transition
distributions for each~$u \in \Nat$
are given by
\[
\pr ( Y (u,t) = i^\prime \given Y (u,t-1) = i )
\sim [P_t]_{i,i^\prime}
\]
i.e., the $(i,i^\prime)$ entry of $P_t$, which is a
random transition matrix drawn from~$\Sigma$.
Let~$\Y^{\star}_{\Sigma}$ denote a discrete-time
process constructed by this procedure with probability
measure~$\Sigma$. By construction, the process is an
exchangeable, Markov multi-state survival process
in discrete time.
Theorem~\ref{thm:discrete_main} states
that this procedure describes
all such processes. 
The proof is left to supplementary section~\ref{section:discrete_proof}.


\begin{theorem}[Discrete-time characterization] \normalfont
  \label{thm:discrete_main}
  Let~$\Y = \{ Y (u, t), u \in \Nat, t \in \Nat\}$ be a
  Markov, exchangeable 
  multi-state survival process.  
  Then there exists a probability measure
  $\Sigma$ on $\P_G$ such that~$\Y_{\Sigma}^\star$ is a 
  version of~$\Y$.
\end{theorem}

\subsection{Continuous-time multi-state survival models}

In continuous-time, the Markov exchangeable
multi-state survival process is governed by a measure on transition
matrices, denoted~$\Sigma$, and a set of constants associated with the edge
set, denoted~${\bf c} = \{ c_{i,i^\prime} \, | \, (i,i^\prime) \in E\}$.
Unlike discrete-time, however, the
jumps occur at random times.  

Consider the $n$-restricted state space process. If $y_{[n]}$ is the
current state then the holding time in this state is exponentially
distributed with a rate parameter depending on the current
\emph{configuration} (see Section~\ref{section:seqdesc}). At a jump
time~$t$, one of two potential events can occur: (a) a \emph{single
  unit}~$u \in [n]$ experiences a transition from $Y(u,t-) = i$ to a
state~$i^\prime \neq i$, or (b) a \emph{subset of~$[n]$} (potentially
a singleton) experience a simultaneous transition. Given the
transition matrix~$P(t)$, these simultaneous transitions are
independent and identically distributed transitions according to~$P(t)
\in \P_G$. The transition matrix~$P(t)$ is obtained from a
measure~$\Sigma$ on $\P_G$ and a set of constants~${\bf c}$; unlike
the discrete-time setting, the measure~$\Sigma$ need not be
integrable. Let~$\Y^{\star}_{\Sigma,c}$ denote a continuous-time
process constructed by this procedure with probability
measure~$\Sigma$. 


\begin{theorem}[Continuous-time characterization] \normalfont
  \label{thm:cts_main}
  Let~${\bf Y} = (\Y (t), t \in \mathbb{R}^+ )$
  be a Markov, exchangeable multi-state survival
  process; and~$I_{s}$ be the $s \times s$ identity
  matrix.
  Then there exists a probability measure~$\Sigma$
  on~$\P_G$ satisfying
  \begin{equation}
    \label{eq:cts_char}
    \Sigma ( \{ I_s \} ) = 0  \text{ and }
    \int_{\P_G} (1-P_{\min} ) \Sigma (dP) < \infty, \quad
    P_{\min} = \min_{i \in [s]} P_{i,i}
  \end{equation}
  and constants~${\bf c} = \{ c_{i, i^\prime} \geq 0 \given 
  (i,i^\prime) \in E \}$
  such that $\Y_{\Sigma,c}^\star$ is a version of~$\Y$.
\end{theorem}

Theorem~\ref{thm:cts_main} generalizes Proposition 4.3 
in~\citet{DempseyMSP} from the simple survival process setting.
The multi-state survival process contains many more 
well-known examples from health and epidemiology --
survival, illness-death, competing risk, and co-morbitity 
processes.  We highlight these examples in supplementary
section~\ref{section:examples}. Within the discussion of examples, we
extend Theorems~\ref{thm:discrete_main} and~\ref{thm:cts_main} to the
setting of recurrent events (see Corollaries~\ref{cor:discrete_recurrent}
and~\ref{cor:cts_recurrent}). 
Here, Theorem~\ref{thm:cts_main} is used to characterize
the continuous-time Markov chain~$\Y_{[n]}$ 
in terms of (1) exponential holding rates 
and (2) transition matrix at jump times. 
We start by defining the {\em characteristic index}
-- a set of functions~$\zeta: [s]^{n} \to \mathbb{R}^+$.

\begin{definition}[Configuration vector]
\label{def:char_index}
\normalfont
For~$y_{[n]} \in [s]^{n}$, define~$x_{[n]} \in [n]^s$ as the \emph{configuration vector} --
an $s$-vector summary of the number of units in
each state. For example, if $s = 2$,~$n = 4$, and
$y_{[4]} = (1, 2, 2, 1)$, then $x_{[4]} = (2,2)$; for~$y_{[4]} = (1,1,2,1)$ 
then~$x_{[4]} = (3,1)$.
\end{definition}

\begin{definition}[Characteristic index]
\label{def:char_index}
\normalfont
The {\em characteristic index}~$\zeta ( y_{[n]} )$ 
is defined as the normalizing constant to the 
integral representation obtained from Theorem~\ref{thm:cts_main}:
\[ 
\zeta_n ( y_{[n]} ) = \int_{\P_G} \left( 1 -
  \prod_{j=1}^s P_{j,j}^{x_j} \right) \Sigma (dP) + 
\sum_{i \in [s]} x_{i} \sum_{i^\prime : (i,i^\prime) \in E} c_{i,i^\prime}
\]
where~$x_{j}$ is the~$j$th component of the configuration vector
associated with~$y_{[n]}$, and the sum is set to $0$ when 
$\{ i^\prime \in V \text{ s.t. } (i,i^\prime) \in E \} = \emptyset$.
Condition~\eqref{eq:cts_char} implies the characteristic index
is finite for any~$y_{[n]} \in [s]^n$.
\end{definition}

For simple survival processes, the characteristic index defined here
simplifies to the characteristic index as defined
in~\citeauthor{DempseyMSP}~\citeyearpar{DempseyMSP}.
At a jump-time~$t$, the probability of transition from~$\Y_{[n]} (t-)$
to~$\Y_{[n]} (t)$ is
\begin{align*}
  q(\Y_{[n]} (t-)&, \Y_{[n]} (t)) = \frac{1}{\zeta_n ( \Y_{[n]} (t-))} \bigg [
  \int_{\P_G} \prod_{u \in [n]} P [ Y (u,t-) , Y (u,t) ] \Sigma (dP) \\
  &+ \delta (\# \{ u \in [n] \given Y (u,t-) \neq Y(u,t) \} = 1 )  
    \sum_{i^\prime : (i, i^\prime) \in E} c_{i,i^\prime} \, \delta (Y
    (u^\prime,t-) = i , Y (u^\prime,t) = i^\prime)  \bigg] \\
                 &=: \frac{\lambda ( \Y_{[n]} (t-), \Y_{[n]} (t))
                   }{\zeta_n ( \Y_{[n]} (t-))}
\end{align*}
where $P[i,i^\prime] = P_{i,i^\prime}$,~$\delta (\cdot)$ is the
indicator function,~$u^\prime$ is the single unit to experience a
transition, and~$\lambda(\cdot,\cdot)$ is the non-normalized transition
function.

\section{Discretization and rounding}
\label{section:discretization}
It has been argued that ``there may be no scientific reason 
to prefer a true continuous time model over a 
fine discretization''~\cite[p. 325]{Breto2009}. 
We tend to disagree with such a viewpoint; 
a basic and very important issue in 
multi-state survival analysis is the distinction 
between inherently discrete data 
(coming from intrinsically time-discrete phenomena) and 
grouped data (coming from rounding of intrinsically continuous data).
Theorems~\ref{thm:discrete_main} and~\ref{thm:cts_main} 
reinforce this distinction
as we see distinct characterizations of 
discrete and continuous-time processes.
One example of the former in survival analysis is time to get pregnant, 
which should be measured in menstrual cycles.
The latter represents the majority of multi-state survival data.
For this reason, we focus the remainder of this paper on the
continuous-time case.  

\section{Description of continuous-time process}
\label{section:desc_cts}

\subsection{Holding times}
Let~$t$ be a jump time at which the state vector~$\Y_{[n]}$ 
transitions into state~$y_{[n]} \in [s]^{n}$. 
To each such state~$y_{[n]}$, we associate an independent exponentially
distributed holding time. By choosing the rate functions in an
appropriate way, the Markov multi-state survival process can be made
both consistent under subsampling, and exchangeable under permutation
of particles.

\begin{corollary} \normalfont
\label{prop:consrates}
A set of rate functions $\{ \tau_n: s^{[n]} \to \mathbb{R}^+ \}_{n=1}^\infty$, 
is consistent if it is {\em proportional to}
the characteristic index $\tau_n (y) \propto \zeta_n (y)$.
\end{corollary}

\noindent Corollary~\ref{prop:consrates} follows from
Theorem~\ref{thm:cts_main}, and shows how the exponential holding rate
relates to the characteristic index; in particular, the difference is
a proportionality constant~$\nu$ which depends on choice of time-scale.

\subsection{Density function}
Since the evolution of the process $\Y_{[n]}$ is Markovian, it is a
straightforward exercise to give an expression for the probability
density function for any specific temporal trajectory.
The probability that the first transition occurs in the interval $dt_1$
with transition from~$\Y_{[n]} (t_1-)$ to $\Y_{[n]} (t_1)$ is
\begin{align*}
&\nu \zeta_n (\Y_{[n]} (t_1-)) \exp \left(- \nu 
\zeta_n (\Y_{[n]} (t_1-)) t_1  \right) dt_1 
\times  q (\Y_{[n]} (t_1-), \Y_{[n]} (t_1)) \\
=&\exp\left(- \nu \zeta_n (\Y_{[n]} (t_1-)) t_1  \right) dt_1 
\times  \lambda (\Y_{[n]} (t_1-), \Y_{[n]} (t_1))
\end{align*}
where $\lambda(\cdot,\cdot)$ is the non-normalized transition probabilities.
Continuing in this way, it can be seen that the joint density for
a particular temporal trajectory $\Y_{[n]}$ consisting of $k$ transitions with
transition times $0 < t_1 < \cdots < t_k$ is
\begin{equation}\label{density}
\exp\Bigl(- \int_0^\infty \nu \zeta_n ( \Y_{[n]} (s) ) \, ds \Bigr)
\prod_{j=1}^k \lambda\bigl(\Y_{[n]} (t_j -), \Y_{[n]} (t_j) \bigr).
\end{equation}
The number of transitions~$k$ is a random variable whose
distribution is determined by (\ref{density}), and hence 
by~$\zeta_n$. Note that with probability one~$T_{i} < \infty$
for each~$i \in [n]$; therefore, the transition at time~$t_k$ 
is such that~$Y_{[n]} (t_k) = (\flat,\ldots,\flat) := \bar \flat$ (i.e., all units
have failed by time~$t_k$). By definition~$\zeta_{[n]} ( \bar \flat )
= 0$ and so the integral in equation~\eqref{density} is finite with 
probability one.

The $n$-dimensional joint distribution is continuous 
in the sense that it has no fixed atoms.
For $n\ge 2$, it is not continuous with respect to 
Lebesgue measure in $\mathbb{R}^n$.
The one-dimensional marginal process is a Markov multi-state survival
process with holding rates $\{ \zeta_1(i) \given i \in [s] \}$,
assuming a valid starting state.  For example, in a survival process
the only valid starting state is ``Alive''.

Although the argument leading to (\ref{density}) did not explicitly
consider censoring, the density function has been expressed in integral
form so that censoring is accommodated correctly.
The pattern of censoring affects the evolution of $\Y_{[n]}$,
and thus affects the integral,
but the product involves only transitions and transition times.
So long as the censoring mechanism is \emph{exchangeability
  preserving}~\citep{DempseyMSP}, inference based on the joint density
given by equation~\eqref{density} is possible. Both simple type I
censoring and independent censoring mechanism preserve exchangeability.

\subsection{Sequential description}
\label{section:seqdesc}
Kolmogorov consistency permits ease of computation for the trajectory
of a new unit~$u^\prime = n+1$ given trajectories for the first~$n$
units~$\Y_{[n]} = \y_{[n]}$. 
The conditional distribution is best described 
via a set of paired conditional hazard 
measures,~$\{ \Lambda^{(c)}_{i,i^\prime}, \Lambda^{(a)}_{i,i^\prime}
\}_{(i,i^\prime) \in E}$.
For~$(i, i^\prime) \in E$, the pair consists of a continuous
component~$\Lambda^{(c)}_{i,i^\prime}$ in addition to an atomic
measure~$\Lambda^{(a)}_{i, i^\prime}$ with positive mass only at
the observed transition times of~$\y_{[n]}$.

For a time~$t$, not a transition time of~$\y_{[n]}$,
consider the new unit transitioning from state~$i$ to $i^\prime$.
Then the continuous component~$\Lambda^{(c)}_{i,i^\prime}$ has 
hazard and cumulative hazard
\begin{align*}
  h_{i,i^\prime}(t) 
  &=  \zeta_n ( \y_{[n+1]} (t-) ) - \zeta_n ( \y_{[n+1]} (t) )  
    := (\Delta \zeta) ( \y_{[n+1]} (t)) \\
  H_{i,i^\prime} (t) &= \int_0^{t} h_{i,i^\prime} (s) ds
\end{align*}
Note that $\zeta_{n+1}(\y_{[n+1]} (t))$ is piecewise constant, 
so the integral is trivial to compute, 
but censoring implies that it is not necessarily constant
between successive failures.

Now let $t$ be an observed transition time (i.e.,
$\y_{[n]}(t-) \neq \y_{[n]}(t)$) and consider the atomic
measure~$\Lambda_{i,i^\prime}^{(a)}$ associated with switching from
state $i$ to $i^\prime$.
At each such point, the conditional hazards has an atom with finite mass
\[
\Lambda^{(a)}_{i,i^\prime}(\{t\}) = \log\frac{\zeta_n(\y_{[n]}(t-))\, q(\y_{[n]}(t-),\y_{[n]}(t))}
{\zeta_{n+1}(\y_{[n+1]}(t-))\, q(\y_{[n+1]}(t-),\y_{[n+1]}(t))},
\]
or, on the probability scale,
\[
\exp(-\Lambda_{i,i^{\prime}} (\{t\})) = \frac{\lambda(\y_{[n+1]}(t-),\y_{[n+1]}(t))}{\lambda(\y_{[n]}(t-),\y_{[n]}(t))}.
\]

The above calculations define the conditional holding time of the new
unit after it enters state~$i$ at time~$t$ (i.e.,~~$Y(n+1,t-) \neq i$ and
$Y(n+1,t) = i$) conditional on~$\Y_{[n]} = \y_{[n]}$. 
For~$s > 0$, let~$\{ t_j \}_{j=1}^L$ denote the observed transition
times of~$y_{[n]}$ within the time-window~$(t, t+s]$. 
Then, the probability that the unit stays in state~$i$ for
\emph{at least} $s>0$ time points is
\[
  \exp \left( -\sum_{i^\prime : (i, i^\prime) \in E} \nu (
    H_{i,i^\prime} (t+s) - H_{i,i^\prime} (t) )
  \right) \cdot 
  \prod_{j=1}^L \exp \left( - \sum_{i^\prime : (i, i^\prime) \in E}
    \Lambda^{(a)}_{i,i^{\prime}} ( \{ t_j \} ) \right).
\]
This conditional probability is used to construct the
Markov Chain Monte Carlo sampling procedure in
Section~\ref{section:estimation}.
For any exchangeable, Markov multi-state survival process with
absorbing state~$\{ \flat_c \}$
such that~$Y_u (0) \not \in  \{ \flat_c \}$ with probability one,
$\zeta_1 (i) > 0$ for all $i \in [s] 
\backslash \{ \flat_c \}$ implies that
the continuous components ($\Lambda^{(c)}_{i,i^\prime}$ for $i \not \in \{
\flat_c \}$) have infinite total mass, so the time until reaching an
absorbing state is finite with probability one.  
This implies, for instance, that in the Aalen comorbidities
study (see Example~\ref{example:multi_risk} in
Section~\ref{section:examples}) that the process will terminate in
death with probability one.

Although the above conditional hazards look complex, it is not
difficult to generate the transition times sequentially for processes
whose characteristic index admits a simple expression for the above
expressions. Right censoring is automatically accommodated by the
integral in the continuous component, so the observed trajectory
$\Y_{[n]}$ may be incomplete.  Below we introduce
the~\emph{self-similar harmonic process} -- a multi-state survival process
which admits such simple expressions. We use this process as a
building block for more complex models in further sections.

\subsection{Self-similar harmonic process}
\label{section:ssharmonic}

Theorem~\ref{thm:cts_main} proves tied failures 
are an intrinsic aspect of Markov multi-state survival processes.
As stated in Section~\ref{section:discretization}, grouped data
usually are the result of rounding of intrisically continous data.
For these processes to be useful in biomedical applications, it is
essential that the model should not be sensitive to rounding.
Sensitivity to rounding is addressed by restricting attention to
processes whose conditional distributions are \emph{weakly
  continuous}, i.e., a small perturbation of the transition times
gives rise to a small perturbation of the conditional distribution.

\citeauthor{DempseyMSP}~\citeyearpar{DempseyMSP} originally studied
this question in the context of exchangeable, Markov survival
processes. In particular, it is shown that the~\emph{harmonic process}
is the only Markov survival process with weakly-continuous conditional
distributions.
Here, we extend the harmonic process to a multi-state survival process
by associating with each edge~$(i,j) \in E$ an independent harmonic
process with parameters~$(\nu_{i,j}, \rho_{i,j})$.
For~$(i,j) \in E$ let~$t_1^{(i,j)} < \ldots < t_{k(i,j)}^{(i,j)}$ denote
the unique observed transition times from~$i$ to~$j$
for~$\Y_{[n]}$ and let~~$\Y_{[n]}^{\sharp}(t; i) = \# \{ u \in [n]
\text{ s.t. } Y_{u} (t) = i \}$; then the continuous component of the hazard
is given by:
\[
  H_{i,j} (t) = \sum_{l: t_l^{(i,j)} \le t} \nu_{i,j} 
  \frac{t^{(i,j)}_{l} -
    t^{(i,j)}_{l-1}}{\Y_{[n]}^{\sharp}(t^{(i,j)}_{l-1}; i) + \rho_{i,j}} 
  + \nu_{i,j}\frac{t - t_{m}^{(i,j)}}{\Y_{[n]}^{\sharp}(t_m^{(i,j)}; i)+ \rho_{i,j}}.
\]
where the sum runs over transition times~$t_l^{(i,j)} \le t$,
and $t_m^{(i,j)}$ is the last such event. The discrete component is a
product over transition times 
\begin{equation}\label{kaplan-meier}
\prod_{l : t_l^{(i,j)} \le t} 
\frac{\Y_{[n]}^{\sharp}(t; i)+\rho_{i,j}}{\Y_{[n]}^{\sharp}(t-; i)+\rho_{i,j}}.
\end{equation}
For small $\{ \rho_{i,j} \}_{(i,j) \in E}$, the combined discrete
components above are essentially the same as the right-continuous
version of the Aalen-Johansen estimator.

We call this process the {\em self-similar harmonic process with
  transition graph~$G$}.
The associated measure~$\Sigma$ on $\P_{G}$ is
\[
\Sigma (dP) = {\delta} [ \# \{ p_{i,j} > 0, (i,j) \in E \} = 1 ]
\, \nu_{\star} (1-p_\star)^{-1} p_{\star}^{\rho_\star - 1} dp_{\star}
\]
where~$p_\star$ is the single non-zero, off-diagonal entry,
$\delta[\cdot]$ is the indicator function, and~$(\nu_{\star},
\rho_{\star})$ are the associated parameters.

While the self-similar harmonic process has strong appeal for use in
applied work, we show below that it is not a universally optimal
choice. The strong assumption embedded in the self-similar harmonic
process is that the transition processes are independent across
edges~$(i,i^\prime) \in E$. 
This implies that at each transition time only transitions along
a single edge $(i,i^\prime) \in E$ are possible. 
We argue below that while this may make sense in specific instances,
additional care is needed in writing down appropriate models for
multi-state survival processes in general. 


\section{Composable multi-state survival models}\label{section:nested}

We now discuss constraints on the multi-state survival models based on 
decompositions of the state-space~$[s]$.
We use the bidirectional illness-death process as an illustrative example.
Recall this process has three states,~\{Healthy, Ill, Dead\} (i.e,
$\{1,2,3\}$); see Example \ref{example:idp} for further details. The
state ``Dead'' is unique and distinct from the
states ``Healthy'' and ``Ill''.  Indeed, the latter two states require
the individual to be categorized more broadly as alive, and are thus
refinements of this more general state (i.e. ``Alive''). 
Suppose the labels ``Healthy'' and ``Ill''  were uninformative with
respect to failure transitions. Then the refinement is immaterial, and
the transition rules should collapse to the transition rule for an
exchangeable, Markov survival process.

The above discussion leads to two natural constraints:
(1) the state ``Dead''~(State 3) is unique and distinct from the other
states,  and (2) the states ``Healthy''~(State 1) and ``Ill''  (State
2) should be considered \emph{partially
  exchangeable}~\citep{Finetti1972}.   To satisfy this,
we require particular constraints on the measure~$\Sigma$ on $3 \times
3$ transition matrices.  First, the measure must only take positive mass
on one of two sets of transition matrices: (a) $P$ with off-diagonal
positive mass only in entries $(1,3)$ and/or $(2,3)$, or (b) $P$ with
off-diagonal positive mass only in entries $(1,2)$ and/or $(2,1)$. The
first set are transition matrices representing transitions from
``Healthy'' or ``Ill'' to ``Dead''. The second are transitions between
``Healthy'' and ``Ill'' or vice versa.
The partition~$B = \{ B_1, B_2 \} = \{ \{1,2\}, \{3\} \}$ splits the
space into two disjoint sets.  Within each block~$B_i$ the states are
partially exchangeable.  We say that the bidirectional illness-death
process is thus relatively partially exchangeable with respect to the
partition~$B$.

Let~$(n_1, n_2)$ denote the number of individuals in states
``Healthy'' and ``Ill'' directly preceding
a transition of type (a). 
Then the probability that $d_1 \leq n_1$ and $d_2 \leq n_2$
individuals respectively transition to the state ``Dead'' is
proportional to~$\int p_{1,1}^{n_1 - d_1} p_{1,3}^{d_1} p_{2,2}^{n_1 - d_1}
p_{2,3}^{d_2} \tilde{\Sigma}(dP)$, where~$\tilde{\Sigma}(dP)$ is the
measure restricted to type (a) transition matrices.  That is,~$\tilde
\Sigma$ puts positive mass on transition matrices~$P$ such
that~$P_{1,2} = P_{2,1} = 0$ so $P_{k,3} = 1-P_{k,3}$ for $k = 1,2$,
and~$P_{3,3} = 1$ (i.e., ``Dead'' is an absorbing state).
In this paper, we focus on the following choice of the restricted
measure:
\[
\tilde{\Sigma}(dP) = \nu_{1,1} \cdot P_{1,1}^{\rho_{1,1} - 1}
(1-P_{1,1})^{-1} \delta( P_{2,2}^{\gamma} = P_{1,1} ) dP_{1,1} dP_{2,2}.
\]
A similar formula exists for transitions of type (b).
The choice corresponds to a proportional model on the
logarithmic scale. It strongly links~$P_{1,1}$ and~$P_{2,2}$ with
baseline measure equivalent to that for a harmonic process.
In supplementary section~\ref{app:loglinear}, we provide more details
on the connection to the proportional conditional hazards model. 
We now consider extending this approach to construction of the
measure~$\Sigma$ for any multi-state survival process.

\subsection{Composable multi-state survival process}
\label{section:composable}

We now generalize the above by introducing $B$-composable processes.

\begin{definition} \normalfont
\label{defn:composable}
A multi-state survival process is \emph{$B$-composable} if
there exists a partition $B=(B_1, \ldots, B_k)$ of the
state-space~$[s]$ such that elements within block~$B_i$ are  
\emph{partially exchangeable} with respect to transition graph~$G$.  
\end{definition}

In other words, a $B$-composable process is any process that is
relatively partially exchangeable with respect to $B$.
For the bi-directional illness-death process (ex.~\ref{example:idp}),
$B = (\{ 1, 2 \}, \{3\})$. 
For comorbidities  (ex.~\ref{example:competing_risk}), $B$ partitions
the risk processes. Two risk processes in the same partition may
represent the same underlying process.
For the competing risks (ex.~\ref{example:competing_risk}), $B =
(\{ 1\}, B_2, \ldots, B_k)$ where~$(B_2,\ldots, B_k)$ partition the
absorbing states and the single state, ``Alive'', is distinct which
implies $B_1 = \{1\}$.
Every process is composable via the degenerate partition where
each~$B_j$ is a singleton and~$k = s$. The partition accounts for
absorbing states,~$i \in \{ \flat_c \}$, by satisfying~$B(i^\prime)
\cap B(i) = \emptyset$ for~$i^\prime \not \in \{ \flat_c \}$, where
$B(j)$ is the component of~$B$ to which state~$j \in [s]$ belongs.

If~$\Y$ is $B^\prime$-composable and $B^\prime$ is a refinement of the partition
$B$, then $\Y$ is also $B$-composable. To avoid confusion, from here on, when 
we say $\Y$ is $B$-composable, we assume there does not exist a
refinement~$B^\prime$ such that $\Y$ is also $B^\prime$-composable.
Definition~\ref{defn:composable} is similar in spirit to that of
\citeauthor{Schweder1970}~\citeyear{Schweder1970}---both aim to
formalize the notion that state changes in the process~$\Y$ are due to
changes in different components. 

\subsection{Choice of measure~$\Sigma$ for a $B$-composable process}

Here, we construct an appropriate measure~$\Sigma$ for a
$B$-composable, Markov, exchangeable multi-state survival process.
The measure will take positive mass only on transitions from states
within~$B_j$ to states within~$B_{j^\prime}$ for a single choice of~$j, j^\prime \in
\{1,\ldots, k \} := [k]$ indexing components of the partition~$B$.
For each component~$B_j$, let~$i(j) \in [s]$ denote a
\emph{representative state}. 
Then, for~$j, j^\prime \in [k]$, define the restricted measure on
transitions from states in~$B_j$ to states in~$B_{j^\prime}$, $\tilde
\Sigma_{j j^\prime} (dP)$, by
\begin{align}
  &\nu_{j, j^\prime} \, 
  P_{i(j), i (j)}^{\rho_{j,j^\prime}
    - 1} \left( 1- P_{i(j), i (j)} \right)^{-1} \label{eq:part1a} \\
  \prod_{l \in B_j \backslash i(j)} 
  &\delta \left[ P_{l,l}^{\gamma_{l,j^\prime}} = P_{i(j), i(j^\prime)} \right] dP_{i(j),
    i(j)} \label{eq:part1b} \\
  \prod_{l \in B_j } &\prod_{m \in B_{j^\prime} \, : \, (l,m) \in E}
    \delta[ y_{l,m} = \alpha_{l,m} ] dy_{l,m} \label{eq:part2}
\end{align}
where~$\gamma_{l, j^\prime} > 0$, $\gamma_{i(j), j^\prime} = 1$,
$\alpha_{l,m} \in [0,1]$, $P_{l,m} \in [0,1]$ and $y_{l,m} \in [0,1]$
such that~$\sum_{m \in B_{j^\prime} \, : \, (l,m) \in E} y_{l,m} = 1$.
Here, the assumption is~$P_{l,m} = (1-P_{l,l}) \cdot y_{l,m}$ for $l
\neq m$, and~$P_{l,l} = P_{i(j), i(j)}^{\gamma_{l,j^\prime}}$. 
Lines~\eqref{eq:part1a} and \eqref{eq:part1b} build the general
measure from a baseline harmonic measure and the assumption of
proportionality on the logarithmic scale for~$P_{l,l}$, where the
proportionality constant depends on $l \in B_j$.
Note~$\gamma_{i(j), j^\prime}$ is set to 1 by design, and so the 
parameters measure risk relative to the chosen \emph{representative
  state}. 
Line~\eqref{eq:part2} addresses the fact that a single state~$l \in
B_j$ can transition to multiple states in $B_{j^\prime}$.
Here, we choose an atomic measure located at $\alpha_{l,m}$, where the
parameters satisfy~$\sum_{m \in B_{j^\prime} \, : \, (l,m) \in E}
\alpha_{l,m} = 1$. These parameters lead to a multinomial
distribution across transitions.  We take this approach because it
leads to simple Gibbs updates of the underlying parameters.
In the example presented in section~\ref{section:example}, Mild CAV
status can transition to either No CAV or Severe CAV status.  This is
captured by presupposing a binomial distribution over these transitions.

Note, for~$B = (\{1\}, \ldots, \{s\})$, the above construction
yields the self-similar harmonic process. For~$B = (\{1,2\}, \{3\})$,
the above construction yields the bi-directional illness-death process
from Section~\ref{section:nested}.

\section{Parameter estimation}\label{section:estimation}


In practice, the patient's health status is typically measured at
recruitment ($t=0$), and regularly or intermittently thereafter while
the patient is under observation.
A complete observation on one patient $(\bft, Y[\bft], V, \Delta)$ consists of
the appointment schedule $\bft$, multi-state process measurements
$Y[\bft]$, and a failure/censoring time~$V$,
and a censoring indicator~$\Delta$. For censored records,~$\Delta = 1$
and the censoring time~$V$ is usually, but not necessarily, equal to
the date of the most recent appointment or end of study.

Here, we assume \emph{non-informative observation times}.  In
particular, given previous appointment times $\bft_{k-1} =
(t_1,\ldots, t_{k-1})$ and observation values $Y[\bft_{k-1}]$, the next
appointment time~$t_k$ satisfies
\begin{equation}
\label{eq:seqind}
t_k \indep Y \given (\bft_{k-1}, Y[\bft_{k-1}]).
\end{equation}
In other words, the conditional distribution of the random
interval~$t_k - t_{k-1}$ may depend on the observed history but not on
the subsequent health trajectory.  
This assumption combined with variational independence implies the
component of the likelihood associated with the appointment schedule
can be ignored for maximum likelihood estimation of parameters
associated with the multi-state survival process.
Under~\eqref{eq:seqind}, we propose a Markov Chain Metropolis Hastings
algorithm for posterior computation.

\subsection{Uniformization}

Uniformization~\citep{Jensen1953, HobolthStone2009} is a
well-known technique for generating sample paths for a given Markov
state-space process.
Take a time-homogeneous, continuous-time Markov process with $s \times
s$ transition matrix~$Q$ (i.e., a $s \times s$ matrix
satisfying~$Q_{ii} = -\sum_{j \neq i} Q_{ij}$ and~$Q_{ij} \geq 0$).
A sample path can be generated via Gillespie's algorithm:
(1) given in state~$i \in [s]$, generate a random exponentially
distributed holding time~$\tilde T$ with rate~$-Q_{ii}$;
(2) after~$\tilde T$ time steps, randomly transition to a new
state~$i^\prime \neq i \in [s]$ with probability~$- Q_{ij} / Q_{ii}$.
Alternatively, a sample path can be generated by uniformization with
parameter~$\Omega \geq \max_i |Q_{ii}|$:
(1) generate sequence of 'potential' transition times~$\bfw$ from a
homogeneous Poisson process with intensity~$\Omega$;
(2) run a discrete-time Markov chain with transition matrix~$B = (I +
\frac{Q}{\Omega})$ on the times~$\bfw$ to generate sequence of states~$\bfs$;
(3) construct~$\bar \bfw$ by removing elements of~$\bfw$ where the
process does not transition to a new state, and construct~$\bar \bfs$
to be the sequence of states at times~$\bar \bfw$. 
The sample path is represented by~$(\bar \bfw, \bar \bfs)$.
Both algorithms generate sample paths from the same Markov state-space
process; however, uniformization is highly adaptable to Markov Chain
Monte Carlo techniques.  See \cite{RaoTeh2013} for an excellent
discussion of uniformization-based MCMC for Markov processes.

\subsection{The MCMC Algorithm}
Direct application of existing Gibbs samplers based on uniformization
to our setting is problematic due to combinatorial growth in the
state-space and time-inhomogeneity of the conditional Markov process. 
Luckily, this approach can be adjusted using the sequential
description from Section~\ref{section:seqdesc} as guidance. In this
section, we derive an MCMC algorithm for posterior computations given
irregularly sampled multi-state survival data. 

\subsubsection{Prior specification and MCMC updates}
We start by specifying priors on the parameters $\{
\nu_{j, j^\prime}, \rho_{j, j^\prime}, \gamma_{l,j^\prime} \}_{l \in
  [s], j,j^\prime \in [k]}$, and $\{\alpha_{l,m} \}_{l, m \in [s]}$. We
write~$\Phi$ to denote the complete set of parameters.
We write $\bar \alpha$, $\bar \gamma$, et cetera to denote
each subset of parameters.
Recall for identifiability reasons~$\gamma_{i(j), j^\prime} = 1$
for~$j, j^\prime \in [k]$.
For all other pairs~$(l,j^\prime)$, we take the prior to be a
log-normal mean-zero distribution~$\log (\gamma_{l,j^\prime}) \sim
N(0, 1)$.  Weakly informative default priors are an
alternative~\cite{Gelman2008}; this corresponds to a log-Cauchy prior
with center 0 and scale parameter~$s_{l,m}$.  However, we saw minimal
differences in simulation studies for our current setting. 
The complete-data likelihood is non-conjugate, so we perform
Metropolis-Hastings updates, conditional on both the complete
trajectory $\Y_{[n]}$ and all other parameters. 

For~$\rho_{j,j^\prime}$, we follow~\cite{DempseyMSP}
and take~$\rho_{j,j^\prime} := \rho$ as a fixed tuning parameter.
Next, define $\lambda_{j,j^\prime} = \nu_{j,j^\prime} \cdot \rho$.
Scaling by~$\rho$ allows direct comparison of~$\lambda_{j,j^\prime}$
across various choices of the tuning parameter.
We choose a Gamma prior, $\lambda_{j,j^\prime} \sim \text{Gamma}
(\alpha, \beta)$, which is conjugate to the complete-data likelihood.
The posterior distribution conditional on the complete trajectory
and other parameters is 
\begin{equation}
\label{eq:samplambda}
\lambda_{j,j^\prime} \given \Y_{[n]}, \bar \gamma, \bar \alpha, \rho
\sim \text{Gamma} \left ( \alpha + k_{j,j^\prime}, \beta + \rho \cdot \int_0^\infty
\zeta_n (Y_{[n]} (s); \bar \gamma, \bar \alpha, \rho ) ds \right )
\end{equation}
where $k_{j,j^\prime}$ is the number of transition between blocks $j$
and $j^\prime$. 
Finally, for~$l \in [s]$ consider transitions to partition
$B_{j^\prime}$.  Index states in $B_{j^\prime}$ such that a transition
from $l$ is possible by $1,\ldots, m_{l, j^\prime}$.  Then the prior for $\bar
\alpha_{l,j^\prime} = (\alpha_{l,1,j^\prime},\ldots,
\alpha_{l,m_{l,j^\prime}, j^\prime})$  is a Dirichlet distribution
with parameters $\bar p_{l, j^\prime} = (p_{l,1,j^\prime}, \ldots,
p_{l, m_{l, j^\prime}, j^\prime} )$.  Then the posterior is conjugate
and
\begin{equation}
\label{eq:sampalpha}
\bar \alpha_{l, j^\prime} \given \Y_{[n]}, \bar \gamma, \bar \lambda, \rho
\sim \text{Dir} \left ( p_{l,1,j^\prime} + k_{l, 1, j^\prime}, \ldots,
p_{l,m_l,j^\prime} + k_{l, m_l, j^\prime} \right )
\end{equation}
where~$k_{l, m^\prime ,j^\prime}$ counts the number of transitions
from state~$l$ to state~$m^\prime$ in $\Y_{[n]}$.

\subsubsection{Conditional sampling patient trajectories}
\label{section:singleunitsample}

We now adapt uniformization to construct a Gibbs sampling of the 
patient trajectory given all other trajectories~$\Y_{-i} = \y_{-i}$,
parameters~$\Phi$, and the previous trajectory for this
patient~$\tilde \y_i$. For patient~$i$, we have observation~$(\bft_i,
\Y_i [ \bft_i ], V_i, \Delta_i)$. The appointment schedule~$\bft_i$ is
an ordered sequence~$0 \leq t_{i,0} < \ldots < t_{i,k_i} \leq V_i$ where
$t_{i,k_i}=V_i$ if and only if $\Delta_i = 1$.
By the Markov property, we can focus on generating the patient
trajectory for fixed parameters for each interval~$[t_{i,j},
t_{i,j+1}]$ separately.
Let~$t_{i,j} \leq \tilde t_{1,j} < \ldots < \tilde t_{L_j, j} \leq
t_{i,j+1}$ be the unique transition times for all other patients
within the interval $[t_{i,j}, t_{i,j+1}]$.
Denote this set of transition times~$\tilde \bft_{j}$. At each time~$t \in
[t_{i,j}, t_{i,j+1}]$, define~$\Omega_t = C \cdot \max_{(i,i^\prime)
  \in E} |\Lambda_{i, i^\prime}^{(c)} (t)|$ for some constant~$C > 1$.
By definition, this is a piecewise constant function that changes at
unique transition times or censoring times.
Next, sample a Poisson process~$\bfw  \subset [t_{i,j}, t_{i,j+1}]$
with piecewise-constant rate $R_t = \Omega_t - \Lambda^{(c)}_{\tilde
  y_i (t), \tilde y_i (t)} (t)$,
where $\Lambda^{(c)}_{\tilde y_i (t), \tilde y_i (t)} (t)$ is the 
continuous component of the conditional distribution.
Finally, let~$\bfu_i$ denote the transition times from the previous
trajectory~$\tilde \y_{i}$.
Then let $\mathcal{T} = \bfw \cup \tilde \bft \cup \bfu_i$ denote the
union of these times.
We can then apply the forward-filtering, backward sampling algorithm
with transition matrix $B_t = (I + \Lambda^{(c)} / \Omega_t)$ at times
$t \in W \cup (U_i \setminus \tilde t)$ and transition
matrix~$\Lambda^{(a)}$ at times~$\tilde t$.  This combines the
standard Gibbs sampling based on uniformization with the added
complexity of the atomic component associated with the conditional
distributions for exchangeable, Markov multi-state survival processes.
 
\subsubsection{MCMC procedure}

In each iteration, we proceed sequentially through patients, sampling
a latent multi-state path for patient~$\Y_i$ given all other latent
processes~$\Y_{-i} := \Y_{[n] \backslash i}$ as discussed in
Section~\ref{section:singleunitsample}.
Then, conditional on the multi-state process~$\Y_{[n]}$, we perform
Metropolis-Hastings updates for $\bar \gamma$ as the complete-data
likelihoods are easy to compute via equation~\eqref{density}. We end
each iteration by using equation~\eqref{eq:samplambda}
and~\eqref{eq:sampalpha} to sample from the posterior for~$\bar
\lambda$ and $\bar \alpha$ respectively.
One issue with this procedure is sequential sampling of the latent
process is computationally expensive and can significantly slow down
the MCMC procedure.  To address this issue, we also propose an
approximate MCMC algorithm in which the latent processes are only
updated every few iterations.  We see performance is not significantly
altered, and run time drops significantly.  
In the illustrative example below, we show $1000$ iterations of the
MCMC procedure where the latent process updates runs every $25$
iterations has similar perforance as the standard MCMC procedure with
significant decrease in overall runtime.

\section{MCMC procedure: a simulation example}
\label{app:gibbsexample}

We illustrate the sampling procedure on a bi-directional
illness-death model (example~\ref{example:idp}).  We set parameters as
follows: first, we assume that marginally a healthy participant
transitions to ill and dead after $2$ and $5$ years respectively (on
average); ill participants to both healthy and dead on average every
$3$ years respectively.  Both healthy and ill participants took on
average $3$ years to transition to failure. We assume a sample of~$N =
250$ individuals, with $150$ initially healthy and $100$ initially
unhealth, were generated. 

First, assume all transitions are observed.  Maximum likelihood
estimation is performed.  
Next, assume the state of each individual is observed annually, with
the transition time to failure observed.
Traceplots in Figure~\ref{fig:traceplots} suggest convergence of the
MCMC procedure after the first 100 iterations. The MCMC sampler gives
posteriors for the parameters. Table~\ref{tab:sim-ex} contains these
estimates.  We see good performance for $\nu_{12}$, $\gamma_{21}$, and
$\gamma_{22}$.  The posterior for $\nu_{11}$ reflects the observation
schedule; indeed, increasing the frequency of observation
significantly improves these posteriors.  In particular, under
complete observations, the posterior distributions are approximately
equal in distribution to the asymptotically normal confidence
intervals.

\begin{table}[!th]
\centering
\begin{tabular}{c c | c c c | c c c}
& & \multicolumn{3}{c}{Maximum Likelihood} 
  & \multicolumn{3}{c}{Posterior distribution} \\ 
Parameter & True Value & Estimate 
& Lower CI & Upper CI & Mean 
& 5\% Quantile & 95\% Quantile \\ \hline
$\nu_{11}$ & 0.50 & 0.53 & 0.43 & 0.64 & 0.15 & 0.09 & 0.23  \\
$\nu_{12}$ & 0.20 & 0.19 & 0.13 & 0.25 & 0.20 & 0.15 & 0.25 \\
$\gamma_{21}$ & 0.70 & 0.75 & 0.64 & 0.86 & 0.85 & 0.67 & 1.07 \\
$\gamma_{22}$ & 1.71 & 1.71 & 1.28 & 2.13 & 1.23 & 1.62 & 2.09 \\ \hline
\end{tabular}
\caption{Parameter estimation}
\label{tab:sim-ex}
\end{table}

Removing the first 100 iterations as burn-in, posterior distributions
are presented in Figure~\ref{fig:posterior}. 
Black curves are the MCMC sampler procedure; red curves are the
approximate MCMC sampling procedure with latent process updates every
$25$ iterations.  We see distributions are approximately equal in all
cases, with largest errors for $\nu_{11}$.  This supports the
aforementioned difficulty in estimating $\nu_{11}$ due intermittent
observations. This is a consequence of the observation schedule being
infrequent compared to the underlying stochastic dynamics.  Complete
observations (i.e., more frequent observations) significantly improves
estimation of $\nu_{11}$.

\begin{figure}[!th]
\centering
\begin{subfigure}{.5\textwidth}
  \centering
  \includegraphics[width=.9\linewidth]{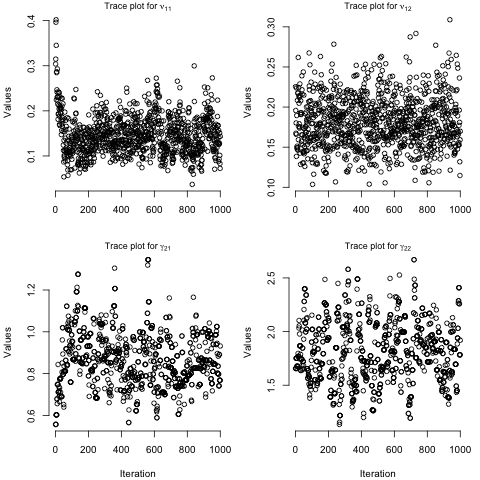}
  \caption{Traceplots for Gibbs sampler}
  \label{fig:traceplots}
\end{subfigure}%
\begin{subfigure}{.5\textwidth}
  \centering
  \includegraphics[width=.9\linewidth]{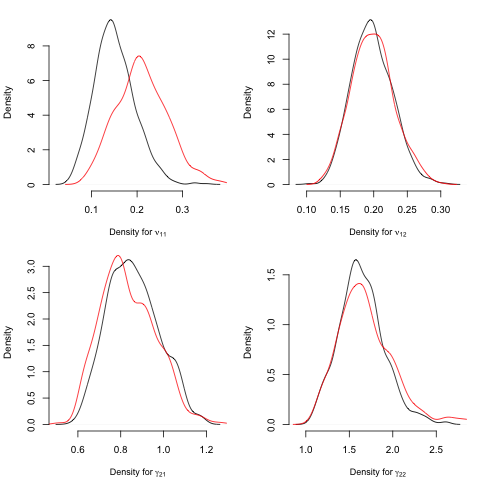}
  \caption{Approximate posterior density}
  \label{fig:posterior}
\end{subfigure}
\caption{MCMC traceplots and densities for simulation example}
\vspace{-5mm}
\end{figure}

Beyond posterior distributions for parameters, one is typically
interested in posterior distributions of the survival functions.
Note, there are two distinct sources of variation -- (1) intermittent
observations and (2) parameter uncertainty.  
The MCMC sampling procedure accounts for both, allowing for survival
functions to be constructed for each iteration of the MCMC sampler
using each iterations' latent process and parameters.
Figure~\ref{fig:survival} presents the pointwise median, $5$\%, and
$95$\% survival at every time since recruitment when the individual is
healthy and ill at baseline respectively.  The red curves are the true
survival function given healthy/ill at baseline.  We see that the
posteriors for survival functions are almost exactly equal, suggesting
intermittent observations did not significantly impact our ability to
predict survival of future patients. 

\begin{figure}[!th]
\centering
\includegraphics[width=.75\linewidth]{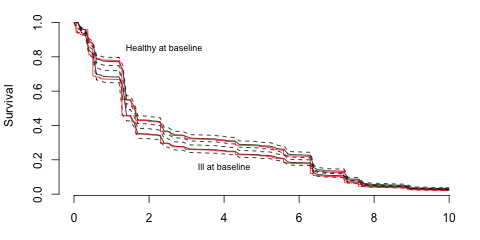}
\caption{Survival functions given baseline state; median (black), 5\%
  and 95\% quantiles (dotted black), and true
  survival function (red)}
\label{fig:survival}
\vspace{-5mm}
\end{figure}

\section{A worked example: Cardiac allograft vasculopathy (CAV) data}
\label{section:example}


To illustrate our methodology, we use data from angiographic
examinations of 622 heart transplant recipients
at Patworth Hopsital in the United Kingdom.
This data was downloaded from the R~library
\url{http://cran.r-project.org/web/packages/msm} maintained by
Christopher Jackson.
Cardiac allograft vasculopathy (CAV) is a deterioration of the
arterial walls.  Four states were defined for heart transplant
recipients: no CAV $(s=1)$, mild/moderate CAV $(s=2)$, severe CAV $(s = 3)$,
and dead $(s = 4)$.  
The transition graph is given by Figure~\ref{fig:cav_example}.
Yearly examinations occurred for up to 18 years following the
transplant. Mean follow-up time, however, is $5.9$ years.
Survival times are observed and/or censored, but CAV state was only
observed at appointment times prior to death/censoring.
For censored recovrds, the censoring time is assumed to be the final
appointment time.  Out of the $622$ patients, Only $192$ patients were
observed in state $2$  (Mild CAV) at any point during their follow-up.
Out of these $192$, $43$ of these patients were subsequently observed
in state $1$.  Only $92$ patients were observed in state $3$ (Severe
CAV) at any point during their follow-up.  Out of these $92$, $12$ of
these patients were subsequently observed in state $2$.  There was no
overlap in the these two patient subsets.

We set~$B = (\{1,2,3\}, \{ 4 \})$, and assume the underlying
process is a~$B$-composable, exchangeable, Markov multi-state
survival process.
Parameters are~$\{ \rho_{(j,j^\prime)}, \lambda_{(j,j^\prime)}
\}_{j,j^\prime=1,2}$, $\{ \gamma_{l,1}, \gamma_{l,2} \}_{l \in B_1}$,
and $\alpha_{2,1} \in [0,1]$.
For identifiability, we set~$\gamma_{1,1}$ and $\gamma_{1,2}$ 
equal to one.
As transitions from state $2$ to $1$ occur but should not occur too
often, the prior on $\alpha = \alpha_{2,1}$ is set to a Beta distribution
with parameters $2$ and $8$ respectively.  
For parameters~$\nu_{11}$ and $\nu_{12}$, the Gamma prior has
hyperparameters $1$ and $1$. 
Parameters~$(\gamma_{21}, \gamma_{22}, \gamma_{31}, \gamma_{32})$ have
independent, standard log-normal priors.
We use the approximate MCMC sampler to perform inference.

\begin{figure}[!th]
\centering
\includegraphics[width=0.7\textwidth]{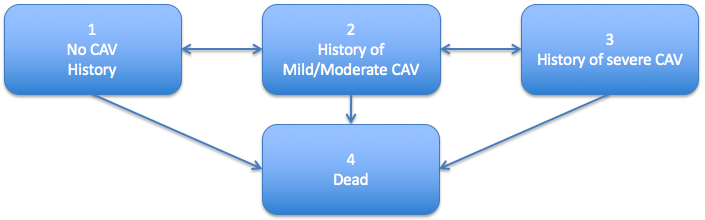}
\caption{Cardiac allograft vasculopathy (CAV) transition
  diagram. 
}
\label{fig:cav_example}
\end{figure}

Traceplots in Figure~\ref{fig:cav-traceplots} suggest convergence
after the first 100 iterations. The MCMC sampler gives posteriors for
the parameters. The posterior mean of $\nu_{11}$ is $0.75$ (i.e.,
marginal time until transition from state $1$ to state $2$ is $1.33$
years). Parameters $(\gamma_{21}, \gamma_{31})$ have posterior means
$(2.43, 0.65)$, translating into marginal holding times of $0.59$ 
and $2.02$ years respectively.  The posterior mean for $\alpha$ is
$0.38$, suggesting that the patient is a bit more likely to experience
progression of the CAV status than regression.
The posterior mean for $\nu_{12}$ is $0.52$ (i.e., marginally the
holding time in state $1$ until a transition to state $4$ is $1.92$
years). This suggests that in state $1$, disease progression is
slightly more likely than failure.
Parameters $(\gamma_{22}, \gamma_{32})$ have posterior means $(1.10, 1.81)$
respectively. This translates marginally into holding times of $1.76$ and
$1.10$ years respectively. For state $2$, the distribution of
$\gamma_{22}$ does overlap $1.0$ suggesting that failure transitions
from states $1$ and $2$ may occur at similar rates.

\begin{figure}[!h]
  \centering
  \begin{subfigure}{.5\textwidth}
    \centering
    \includegraphics[width=.95\linewidth]{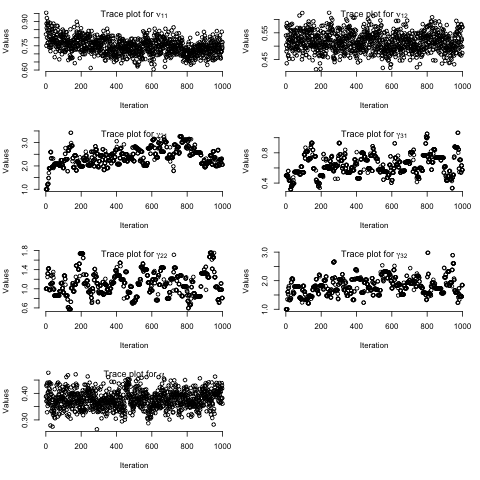}
    \caption{Traceplots for MCMC algorithm}
    \label{fig:cav-traceplots}
  \end{subfigure}%
  \begin{subfigure}{.5\textwidth}
    \centering
    \includegraphics[width=.95\linewidth]{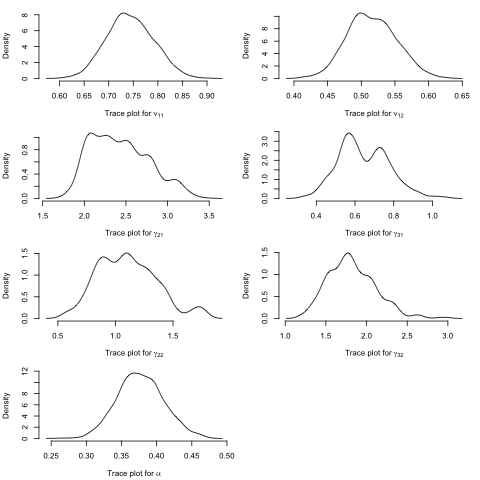}
    \caption{Approximate posterior densities}
    \label{fig:cav-quantile}
  \end{subfigure}
  \caption{MCMC traceplots and densities for CAV study}
  \vspace{-5mm}
\end{figure}

We next consider the posterior distributions for the survival
functions.
Figure~\ref{fig:cav-survival} plots median survival at each time $t$ over
all iterations of the MCMC sampler as well as the Kaplan-Meier
survival function estimator.
Note at baseline, all patients are in state $1$; therefore, the
Kaplan-Meier curve should be compared with the median survival curve
given the new patient is in state~$1$.  We see that the posterior
survival curve is significantly lower.  This reflects expected disease
progression since baseline; that is, a new patient in state $1$ at
baseline who is alive at time $t$ is likely to have seen a progression
in their disease state, leading to an increase in their current risk
of failure.  Under the multi-state survival process model, the
expected survival time from baseline given new patient is in state
$1$, $2$, and $3$ is $8.75$, $8.32$, and $7.32$ respectively. 
Under the Kaplan-Meier estimator, the expected survival time from
baseline is $9.66$.
We have included the 5\% and 95\% quantiles for the survival function
at each time $t$ when the patient is in state $3$ at baseline.
Figure~\ref{fig:cav-survival-at-five} plots median survival over all
iterations of the MCMC sampler given the user is alive at time $t=5$.
Under the multi-state survival process model, the expected
survival time from time $5$ given new patient is in state $1$, $2$,
and $3$ is $5.97$, $5.55$, and $5.04$ respectively. 
We again include the 5\% and 95\% quantiles for the survival function
at each time $t$ when the patient is in state $3$ at baseline.

\begin{figure}[!h]
  \centering
  \begin{subfigure}{.5\textwidth}
    \centering
    \includegraphics[width=.95\linewidth]{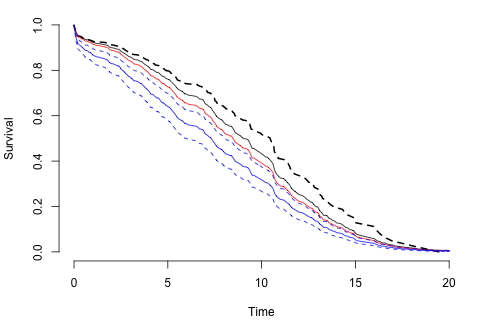}
    \caption{Survival functions at baseline}
    \label{fig:cav-survival}
  \end{subfigure}%
  \begin{subfigure}{.5\textwidth}
    \centering
    \includegraphics[width=.95\linewidth]{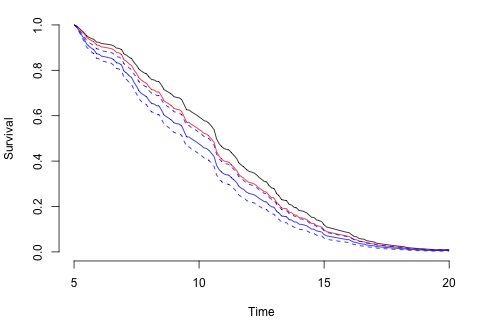}
    \caption{Survival functions at $t = 5$}
    \label{fig:cav-survival-at-five}
  \end{subfigure}
  \caption{Survival functions given ``No CAV''  (black),
    ``Mild/Moderate CAV'' (red), and ``Severe CAV'' (blue);
    Kaplan-Meier estimator (dotted black).} 
  \vspace{-5mm}
\end{figure}

\section{Concluding remarks}\label{section:concluding remarks}

The preceding pages lay a theoretical and methodological foundation
for the development of exchangeable, Markov multi-state survival
processes. The approach put forward sits between the prior parametric
and non-parametric approaches for multi-state survival data.  
A representation theorem characterizes the entire space of
exchangeable, Markov multi-state processes via a measure $\Sigma$ and
constants ${\bf c}$.
Restrictions on choice of measure were driven by prior work on weak
continuity and the notion of composable systems developed in this
paper.  To fit the models to intermittently observed multi-state
survival data, extensions of existing MCMC samplers for Markov jump
processes were required; an approximate MCMC sampler version was
derived that achieved good empirical performance in simulations and
led to significant decreases in runtime. 
Application to cardiac allograft vasculopathy (CAV) data showed how to
interpret posterior parameter distributions.  Comparison to
Kaplan-Meier estimators showed how the model adjusts the
non-parametric survival function estimator to account for disease
progression.

\bibliography{statespace-refs}
\bibliographystyle{plainnat}

\appendix

\section{Proof of discrete-time characterization}
\label{section:discrete_proof}

We start with a discussion of an equivalent
matrix representation used in proofs of
both discrete and continuous-time characterizations.

\subsection{Matrix equivalent representation}

For every~$y \in [s]^{\mathbb{N}}$, there exists 
an equivalent representation as a matrix with
an infinite number of rows and $k$ columns~$M
\in [s]^{\mathbb{N} \otimes k}$ where the first
row~$M_{1,\cdot} = [y_1, \ldots, y_k ]$.
Let~$M_{\cdot,i} \in [s]^{\mathbb{N}}$ denote the $i$th column of 
$M$ (i.e., $M = [ M_{\cdot,1} \, | \, \ldots \, | \, M_{\cdot,k} ]$.
Exchangeability of~$\Y$ implies 
column exchangeability of~$\M$.  That is,
for a set of permutations~$(\sigma_1, \ldots, \sigma_k)$ 
such that~$\sigma_i: \mathbb{N} \to \mathbb{N}$ for $i = 1,\ldots,k$,
we have
\[
\M = [ \M_{\cdot, 1} \, | \, \ldots \, | \, \M_{\cdot, k} ] \overset{D}{=} 
[\M_{\cdot, 1}^{\sigma_1} \, | \, \ldots | \, \M^{\sigma_k}_{\cdot, k} ] = \M^{\sigma}
\]
where $\overset{D}{=}$ stands for equivalent in distribution.
We define this property as {\em column-wise exchangeable}.
Note exchangeability implies column-wise exchangeability but not
vice versa.
Restriction acts column-wise
\[
\M^{[n]} = [ \M^{[n]}_{\cdot, 1} \, | \, \ldots \, | \, \M_{\cdot, k}^{[n]} ]
\]
with such matrices in one-to-one correspondence with
elements in~$[s]^{[n \cdot k]}$.

We define an action of the matrix representation~$A$ on $y \in [s]^{\mathbb{N}}$
by $A(y) = (A_{1,y_1}, A_{2,y_2}, \ldots)$. In other words,
the $i$th row of $A$,~$A_{i,\cdot}$, acts on $y_i$ by sending it to $A_{i,y_i}$.
The identity map~$I$ is defined by each row~$I_{i, \cdot}$ being 
equal to $[1 2 \ldots k]$; then $I(y) = y$ for all $y \in [s]^{\mathbb{N}}$.
The equivalent vector representation of $I$ is defined as~$\text{id} \in [s]^{\mathbb{N}}$.

We express the {\em asymptotic frequency} of~$A$ by $k$-vector
$|A|_k = (|A_1|,\ldots, |A_k|)$ assuming $|A_i|$ exists.  

The proofs below are for the complete graph case. As~$G$
is simply a restriction of the measure to a particular
subset of transition matrices~$\P_G$, the proofs below
yield the desired results.

\begin{proof}[Proof of Theorem~\ref{thm:discrete_main}]
  By Kolmogorov consistency,~$\Y_{[n]}$ is a Markov chain
  governed by transition probability rules~$\pr ( \Y (t) = y^\prime 
  \given \Y (t-1) = y)$. 
  Restriction to~$[n]$ yields a transition rule for $y, y^\prime \in [s]^{[n]}$:
  \[
    \pr_{n} ( \Y_{[n]} (t) = y^\prime \given \Y_{[n]} (t-1) = y ) = 
    \pr ( \Y (t) = R_n^{-1} (y^\prime) \given \Y (t-1) = y^\star )
  \]
  where $R_n$ is the restriction operation so 
  $R_n^{-1} (y^\prime) = \{ y \in [s]^{\mathbb{N}} 
  \text{ s.t. } R_n (y) = y^\prime \}$
  and $y^\star \in R_n^{-1} (y)$.
  Without loss of generality, 
  we focus on time~$t=1$.

  We define a measure~$\eta$ by
  \[
    \eta (\cdot) := \pr ( \cdot \given \text{id} ) 
  \]
  Via the matrix representation, we can think of~$\eta$ as a 
  measure on matrices~$A \in [s]^{\mathbb{N} \otimes k}$.
  Restriction to~$[n]$ yields $A^{[n]} \sim \eta^{[n]} (\cdot) = 
  \pr_{nk} (\cdot \given \text{id}_{nk})$.
  The action of $A^{[n]}$ on $x \in [s]^{[n]}$ is then given by
  \[
    A^{[n]} (x) = \eta^{[n]} (x) = \pr_{n} ( \cdot \given I_n (x)) = \pr_{n} ( \cdot \given x )
  \]
  as we require.

  The above argument shows that there exists a measure~$\eta$ 
  such that $\Y^\star$ defined by
  \[
    \Y^\star (t) = ( A_t \circ A_{t-1} \circ \ldots A_{1} ) (Y_0^\star)
  \]
  is equivalent in distribution to~$\Y$.
  Here $A_t$ are independent, identical distributed draws from $\eta$
  for each time~$t \in \mathbb{N}$.

\end{proof}

\begin{proof}[Proof of Corollary~\ref{cor:discrete_recurrent}]
  Consider the recurrent event process~$\Y$ up
  time $\tau<\infty$. Then $\Y^\star = \tau \wedge \Y$
  is a version of $\Y$ on $t \in 0,1,\ldots,\tau$.
  Let~$\eta_{\tau}$ to denote the measure associated
  with~$\Y^{\star}$.  
  For $P \in \P_{\tau}$, let $R_{\tau^\prime, \tau} (A)$ be
  the restriction of this $\P_{\tau^\prime}$.
  Then for $\tau^\prime < \tau$
  \[
    \eta_{\tau} ( \{ P \in \P_{\tau} \given R_{\tau^\prime, \tau} (P) = P^{\star} \})
    = \eta_{\tau^\prime} ( \{ P^{\star} \} )
  \]
  So we have consistency across~$\tau > 0$.  
  We define the measure~$\eta$ on $\P_{\infty}$
  by
  \[
    \eta (\cdot) = \lim_{\tau \uparrow \infty} \eta_{\tau} (\cdot)
  \]
  is the unique measure such that
  $\Y^\star$ is a version of $\Y$.
\end{proof}

\section{Proof of continuous-time characterization}
\label{section:continuous_proof}

Again the proof below is for the complete graph case. As~$G$
is simply a restriction of the measure to a particular
subset of transition matrices~$\P_G$, the proof below
yields the desired result.

\begin{proof}[Proof of Theorem~\ref{thm:cts_main}]
  \label{proof:cts_main}
  Like in the discrete-case, we construct the measure~$\eta$
  from the transition rule which governs~$\Y$.
  This will connect~$\Y^{\star}$ to $\Y$ such that they are 
  equal in law.
  
  Since~$\Y_{[n]}$ is a Markov process on $[s]^{[n]}$, it is 
  governed by a transition rate function
  \[
    Q_n (y, y^\prime) = \lim_{t \downarrow 0} \frac{1}{t} \pr ( \Y_{[n]} (t) = y^\prime \given \Y_{[n]} (0) = y).
  \]
  We start by describing the key characteristics of the transition rate function
  \begin{enumerate}
  \item \label{p1} The transition rate function exhibits {\bf finite activity}:
    \[
      \sum_{y^\prime \neq y} Q_n (y, y^\prime) < \infty
    \]
  \item \label{p2} The transition rate function is {\em exchangeable}.  That is,
    for any~$\sigma : \mathbb{N} \to \mathbb{N}$ and $y \neq y^\prime$:
    \[
      Q_n (y, y^\prime) = Q_n (y^{\sigma}, (y^\prime)^{\sigma}).
    \]
  \item \label{p3} The transition rate functions are {\em consistent}. That is,
    for $y, y^\prime \in [s]^{[m]}$ and $m \leq n$, 
    \[
      Q_m (y, y^\prime) = Q_n ( y^\star, R_{m,n}^{-1} (y^\prime)) 
    \]
    where $R_{m,n}^{-1}$ is the inverse of the restriction operator from $[n]$ to $[m]$
    and $y^\star \in R_{m,n}^{-1} (y)$.
  \end{enumerate}

  We then define the measure for $A \in [s]^{[n] \times s} \backslash \{ \id_{k,n} \}$
  as
  \[
    \eta_{n} (A) = Q_n ( \id_{k,n}, A) 
  \]
  This measure is is column-wise exchangeable by exchangeability of~$Q$ and satistfies
  \begin{equation} \label{eq:}
    \eta_n (A) = \eta \left( \{ A^\star : [s]^{[n] \times s} \given (A^{\star})^{[n]} = A \} \right)
  \end{equation}
  for all $m \leq n$ and $A \in [s]^{\mathbb{N} \times s}$ by consistency of~$Q$.

  The measure~$\eta (A) = Q_n ( \id_{k}, A )$ is also column-wise exchangeable and 
  satisfies
  \begin{itemize}
  \item $\eta ( \{ \id_k \}) = 0$ (i.e., a transition must occur) and
  \item $\eta ( \{ A \given A^{[n]} \neq \id_{k,n} ) < \infty$ (i.e., finite, restricted activity)
  \end{itemize}
  Following~\cite{Pitman03}, we construct a process $\Y^\star = ( \Y^\star (t), t \geq 0)$
  via its finite restrictions $\Y_{[n]}^\star = ( \Y_{[n]}^\star (t), t \geq 0)$.
  Let~${\bf A} = \{ (t, A) \subset \mathbb{R}^+ \times [s]^{\mathbb{N} \otimes k} \}$ be a Poisson
  point process with intensity~$dt \otimes \eta$.  
  Given an inital state~$\Y^\star (0)$, for each $t > 0$ if $t$ is an atom of ${\bf A}$ then
  \begin{itemize}
  \item if $A_t^{[n]} \neq \id_{k,n}$, then set~$\Y_{[n]}^\star (t) = A_t^{[n]} (\Y_{[n]}^\star (t-)) $
  \item otherwise~$\Y_{[n]}^\star (t) = \Y_{[n]}^\star (t-)$
  \end{itemize}

  The difference between the continuous and discrete-time setting is the random
  time between jumps and that the jumps (1) occur for an infinite fraction 
  of the units as $n \to \infty$, or (2) occur for a single unit~$u \in \mathbb{N}$.
  By construction for $m \geq n$, the restriction of $Y_{[m]}^\star$ to $[n]$ is consistent with 
  $Y_{[n]}^\star$ so we have $\Y^\star$ is a unique $[s]^{\mathbb{N}}$-valued process.

  \begin{lemma}
    \label{lemma:ystar}
    The process~$\Y^\star$ is a Markov, exchangeable state-space process.
  \end{lemma}
  
  \begin{proof}
    Consistency is given by the above argument; the process is Markovian
    by construction and the assumptions on~$\Y$.  Exchangeability is due
    to $\eta$ being a column-wise exchangeable measure since $\eta_{[n]}$ 
    are finite, column-wise exchangeable measures.
  \end{proof}

  The final concern before showing that 
  $\Y^\star$ is stochastically equivalent
  to $\Y$ is the uniqueness of the 
  measure~$\eta$ related to the restricted
  measures~$\eta_{n}$.  

  \begin{lemma}
    There exists unique measure~$\eta$
    on $[s]^{\mathbb{N} \otimes s}$ 
    which satisfies (1) $\eta (\{\id_k\}) = 0$, 
    (2) $\eta ( \{ A \in [s]^{\mathbb{N} \otimes s} 
    \given A^{[n]} \neq \id_{k,n} \}) < \infty$,
    and
    \begin{equation}
      \label{eq:lemma_condition}
      \eta ( \{ A^{\star} \in [s]^{\mathbb{N} \otimes s} 
      \given (A^{\star})^{[n]} =  A \} ) = \eta_n (A)
    \end{equation}
    for all $n > 0$ and $A \in [s]^{[n] \otimes s}$.
  \end{lemma}

  \begin{proof}
    The sets
    \[
      \{ A^{\star} \in [s]^{\mathbb{N} \otimes s} 
      \given (A^{\star})^{[n]} =  A \}
    \]
    are a $\pi$-system generating the
    $\sigma$-field on $[s]^{\mathbb{N} \otimes s}$.
    The above discussion proves equation
    \ref{eq:lemma_condition}; and the 
    measure is additive.  Therefore,
    uniqueness is a consequence of any measure extended to a
    $\sigma$-algebra being unique if the measure is $\sigma$-finite.
  \end{proof}

  \begin{lemma} 
    $\Y^\star$ is a version of $\Y$. 
  \end{lemma}

  \begin{proof}
    First, the finite restrictions~$\eta_n$
    satisfies 
    \begin{align*}
      \eta_{n} ( \{ A \in [s]^{[n] \otimes s}  
      \given A(y) = y^\prime \} ) 
      &= \sum_{A : A(y) = y^\prime} Q_{nk}( I_{k,n}, A) \\
      &= Q_n ( I_{k,n} (y), A(y)) = Q_n ( y, y^\prime).
    \end{align*}
    And has finite activity:
    \[
      \eta_{n} ( \{ A \in [s]^{[n] \otimes s}  
      \given A(y) = y \} ) = 
      \sum_{y^\prime \neq y} Q_n (y, y^\prime) < \infty
    \]
    Therefore~$(\Y^{\star})^{[n]}$ is an 
    Markov, exchangeable process with
    jump rates~$Q_n (\cdot, \cdot)$.  
    By Kolmogorov's extension theorem,
    the unique process~$\Y^\star$ is a 
    version $\Y$.
  \end{proof}

  We still need to show~$\eta$ can 
  be decomposed into the respective 
  components:
  \begin{itemize}
    \item {\bf Dislocation measure}:
      measure on~$s \times s$ 
      transition matrices$\sim \Sigma$
      which satisfies
      \begin{align*}
        \Sigma ( \{ I_k \} ) &= 0 \\
        \int_{\P_k} (1-P_{\min}) \Sigma (dP) &< \infty
      \end{align*}
      where $P_{\min} = \min_i P_{i,i}$.
    \item {\bf Erosion measures}:
      Let~$A \in [s]^{\mathbb{N} \times s}$ 
      then we call this~$A = \id$ and
      we flip a single unite~$u \in \mathbb{N}$
      (i.e., $A_{u,i} = i^\prime$).
      Let~$\mu^u_{i,i^\prime}$ be this 
      point mass measure. Define
      \[
        \mu_{i,i^\prime} = \sum_{u \in \mathbb{N}} 
        \mu^u_{i,i^\prime}
      \]
      \item The combined measure is given by:
        \[
          \eta_{\Sigma, c} (\cdot)  = \mu_{\Sigma} (\cdot) + \sum_{i
            \neq i^\prime \in [s]} c_{i i^\prime} \mu_{i,i^\prime} (\cdot)
        \]
  \end{itemize}

  \begin{lemma}
    The measure~$\eta_{\Sigma, c}$ is a column-wise exchangeable measure 
    satisfying the necessary constraints.  
  \end{lemma}

  \begin{proof}
    We prove this for each component of~$\eta_{\Sigma,c}$.
    First, $\mu_{\Sigma} (\{ \id_k \}) = 0$ by construction.  Moreover,
    for $P \in \P_k$
    \begin{align*}
      \mu_P ( \{ A \given A^{[n]} \neq \id_{s,n} \} ) &\leq \mu_P ( \{ A \given A^{[n]} \neq \id_{s,n}) \\
                                                      &\leq \sum_{j=1}^{s} \mu_P ( \{ A \given A^{[n]}_{u, j} \neq j \text{for all } u \in [n] \} ) \\
                                                      &\leq k ( 1-p_{\min}^n) \leq n \cdot k ( 1-p_{\min}^n)
    \end{align*}
    which implies
    \[
      \mu_{\Sigma} (\{ A \given A^{[n]} \neq \id_{s,n} \} ) \leq n \cdot k \int_{\P_k} (1-p_{\min}) \Sigma (dP) < \infty
    \]
    by the above assumptions.
    
    Second, $\mu_{i,i^\prime} (\{ \id_k \}) = 0$ by construction. Moreover,
    \begin{align*}
      \sum_{i \neq i^\prime \in [s]} c_{i i^\prime} \mu_{i, i^\prime} ( \{ A \given A^{[n]} \neq \id_{s,n} \})
      &\leq c_{\max} \sum_{i \neq i^\prime \in [s]} \sum_{u \in [n]} \mu^{u}_{i, i^\prime} ( \{ A \given A^{[n]} \neq \id_{s,n} \}) \\
      &\leq c_{\max} {s \choose 2} n < \infty. 
    \end{align*}
  \end{proof}

  So we have that $\Y^\star$ is a version of $\Y$ and for any $\Sigma$ and $c$ the measure is $\mu_{\Sigma, c}$ 
  is column-wise exchangeable satisfying necessary constraints.  It rests to connect show that the measure
  $\eta$ can be decomposed such that there exists $\Sigma$ and $c$ such that $\eta = \mu_{\Sigma,c}$.

  \begin{lemma}
    For $\eta$ constructed from $Q$, $\eta$-almost every 
    $A \in [s]^{\mathbb{N} \otimes s}$ possesses asymptotic frequencies
    $|A|_s \in \P_s$.
  \end{lemma}

  \begin{proof}
    The $\eta$ by construction satisfies the necessary conditions. We set~$\eta_n = \eta$
    on $\{ A \given A^{[n]} \neq \id_{s,n} \}$.  Then $\eta_n$ is column-wise exchangeable
    for $(\sigma_1, \ldots, \sigma_s)$ such that $\sigma_i : \mathbb{N} \to \mathbb{N}$
    fixes $[n]$. 
 
    We can find a column-wise exchangeable measure by simply
    considering ignoring the first~$n$ rows. Let~$\eta_{n}^\prime$ be measure
    obtained from $\eta$ after applying the {\em $n$-shift function} $\phi_n: A \to A^\prime$.
    Then $\eta_n^\phi$ is column-wise exchangeable and therefore has asymptotic frequencies.
    But asymptotic frequencies only depend on such an $n$-shift for every fixed~$n>0$
    (i.e., $|A|_s = |\phi_n (A)|_s$); therefore, $\eta_n$-almost every $A$ has 
    asymptotic frequencies.  

    To prove~$\eta$-almost every $A$ has asymptotic frequencies we simply note that
    $\eta_n \uparrow \eta$ and therefore the monotone convergence theorem
    completes proof.
  \end{proof}

  \begin{lemma}
    There exists a measure~$\Sigma$ such that 
    the restriction of~$\eta$ to $\{ |A|_s \neq I_s \}$ is equivalent to~$\mu_{\Sigma}$.
  \end{lemma}

  \begin{proof}
    Let~$\phi_n(A)^{[m]}$ denote the restriction of $\phi_n (A)$ to $[m]$.
    Then
    \begin{align*}
      \eta_n ( \{ \phi_n(A)^{[2]} \neq \id_{s,2} \} \given |A|_s = P )
      &= \eta^\phi_n ( \{ \phi_n(A)^{[2]} \neq 
        \id_{s,2} \} \given |A|_s = P ) \\
      &= \eta^\phi_n ( \{ A^{[2]}  \neq \id_{s,2} \} \given |A|_s = P ) \\
      &= \mu_P ( \{ A^{[2]}  \neq \id_{s,2} \} ) \\
      &\geq 1 - p_{\min}^2 \geq 1 - p \\
      \Rightarrow \eta_n ( \{ A^{[2]}  \neq \id_{s,2} \} ) 
      &\geq \int_{\P_k} ( 1 - p_{\min} ) \Sigma_n(dP)
    \end{align*}
    with $\Sigma_n = \eta_n {\bf 1}[ |A|_k \neq I_k]$.
    As $n \to \infty$ this yields,
    \[
      \infty > \eta ( \{ \phi (A)^{[2]} \neq \id_{s,2} \} )
      \geq \int_{\P_s} (1-p_{\min}) \Sigma (dP)
    \]
    and $\Sigma ( \{ \id_s \}) = 0$ by construction.
    
    It rests to show that
    $\mu_{\Sigma} = {\bf 1} [ |A|_s \neq \id_s ] \eta$.
    We have 
    \[
      \eta ( \{ A^{[n]} = A^\star, |A|_s \neq \id_{s} \} ) =
      \lim_{m \uparrow \infty}  \eta_m ( \{ A^{[n]} = A^\star,
      A^{[m]} \neq \id_{s,m}, |A|_s \neq \id_{s} \} )
    \]
    The right hand side is equivalent to
    \begin{align*}
      \eta^\phi_m ( \{ A^{[n]} = A^\star,
      |A|_s \neq \id_{s} \} )
      &= \int_{\P_s} \mu_P ( \{ A^{[n]} = A^\star \} ) 
      |\eta_m^\phi| ( |A|_s \in dP ) \\
      &= \int_{\P_k} \mu_P ( \{ A^{[n]} = A^\star \} ) \Sigma (dP) \\
      &= \mu_{\Sigma} ( \{ A^{[n]} = A^\star \} )
    \end{align*}
  \end{proof}

  \begin{lemma}
    There exists a set of constants~$\{ c_{i,i^\prime} \}_{i \neq i^\prime \in [s]}$ 
    such that the restriction of~$\eta$ to $\{ |A|_s = I_s \}$ is equivalent
    to~$\mu_{c}(\cdot) = \sum_{i\neq i^\prime \in [s]} 
    c_{i i^\prime} \mu_{i,i^\prime} (\cdot)$.
  \end{lemma}

  \begin{proof}
    We restrict our attention to the set of $A$ where
    $A^{[2]} \neq \id_{s,2}$ but $\phi_3(A) = \id_{s}$.
    This set~$B$ contains all single unit transition.
    As the measure $\eta_3^\phi$ is proportional
    to the point mass at $\id_{s}$, then $\eta$ restricted
    to the event $\{ A^{[2]} \neq \id_{s,2}, 
    \phi_3(A) = \id_{s}, |A|_s = \id_{s} \}$ is the
    sum
    \[
      \sum_{A \in B} c_{A} \delta_{A} (\cdot) .
    \]
    If $A$ contains more than a single unit transition, 
    exchangeability forces $c_A = 0$ since 
    $\eta ( \{ A \given A^{[2]} \neq \id_{s,2} \} ) < \infty$.
    The same argument shows $A \in [s]^{\mathbb{N} \otimes s}$ 
    such that $|A|_s = I_s$ and $c_A > 0$ implies
    $A$ is a single unit transition.
  \end{proof}

  This concludes the proof.

\end{proof}

\section{Examples}\label{section:examples}

Here, we describe several important examples that motivate the
current study of multi-state survival processes.  

\begin{example}[Survival process]
\label{example:msp} \normalfont
A survival process has state space \{~Alive, Dead \}
with transitions governed by the simple graph shown in Figure~\ref{fig:dead_alive}.
In this case, $s = 2$ and the edge-set is the singleton~$\{(1,2)\}$;
because the state ``Dead'' is absorbing, the space~$\P_G$
is equivalent to the one-dimensional space~$p \in [0,1)$.  
Restricting to $[n]$, suppose that all individuals at time~$t$
are still at risk.  In discrete-time, the probability of~$d$ 
individuals passing away between times~$t$ and~$t+1$ 
is equal to
\[
\int_0^1 p^{n-d} (1-p)^{d} \Sigma (dp)
\] 
where~$\Sigma$ is a probability measure on~$(0,1]$. 
The marginal distribution of the survival time 
for each patient is geometric.  
Letting $\Sigma (dp)$ be the 
conjugate prior~$\nu \cdot p^{\alpha - 1} (1-p)^{\beta-1} dp$
with $\alpha, \beta > 0$ yields a 
discrete-version of the ``beta-splitting'' process~\citep{Aldous1996}.
The marginal geometric distribution in this case
has parameter~$\beta/ (\alpha + \beta)$.

\begin{figure}[!h]
  \begin{center}
    \includegraphics[scale=0.5]{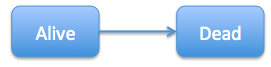}
    \caption{Graph representation of survival process}
    \label{fig:dead_alive}
  \end{center}
  \vspace{-5mm}
\end{figure}

In continuous-time, the probability of~$d$ 
individuals passing away between times~$t$ and~$t+1$ 
is proportional to
\[
\int_0^1 p^{n-d} (1-p)^{d} \Sigma (dp) +  \delta (d = 1) c_{1,2}
\] 
where~$\Sigma$ is a measure on~$(0,1]$ satisfying
$\int (1-p) \Sigma (dp) < \infty$.
In continuous-time, the marginal distribution is exponential.
The conjugate prior now relaxes the constraints to $\beta > -1$.
Considering choice of measure, \citeauthor{DempseyMSP}~\citeyear{DempseyMSP}
suggest choosing measure with~$\beta = 0$ -- called the {\em harmonic
  process}.  The harmonic process is the only family of Markov
survival processes with weakly continuous predictive distributions --
a key property in applied work. The chance of singleton events is set
to zero (i.e., $c_{1,2} = 0$). 
\end{example}

\begin{example}[Illness-death process]
\label{example:idp} \normalfont
The illness-death process has state space~\{~Healthy,~Unhealthy,~Dead\}
with transitions governed by the simple graph shown in Figure~\ref{fig:illness_death}.
The state ``Dead'' (i.e., $s = 3$) is absorbing, the space~$\P_G$
is equivalent to a three-dimensional space.
The \emph{bi-directional} illness-death process includes
the additional edge~(Unhealthy,Health), allowing the patient
to recover.  
Both processes can be viewed 
as refinements of the survival process.

\begin{figure}[!h]
  \begin{center}
    \includegraphics[scale=0.5]{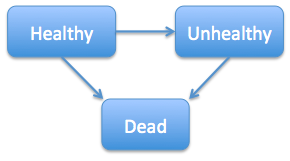}
    \caption{Graph representation of the illness-death process}
    \label{fig:illness_death}
    \vspace{-5mm}
  \end{center}
\end{figure}


An issue arises for the bi-directional illness-death process
when the definitions of ``Healthy'' and ``Unhealthy'' are arbitrary
(i.e., have no scientific value).  Potentially the labels are
exchangeable and, if so, the process is a Markov exchangeable survival
process.   Such considerations lead to natural constraints on the
choice of measure -- see section~\ref{section:nested} for a
discussion.
\end{example}

\begin{example}[Comorbidities]\label{example:multi_risk}
\normalfont
Comorbidities are multiple stochastic processes experienced simultaneously 
by the same patient.
Figure~\ref{fig:multiple_risks}, for example, represents~$L$ binary 
risk processes each with an absorbing state.   In general,~$Y (i,t) =
(Y_1 (i,t), \ldots, Y_L (i,t))$ is an $L$-vector state-space process.

\begin{figure}[!h]
  \begin{center}
    \includegraphics[scale=0.5]{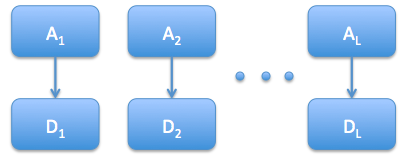}
    \caption{Graph representation of co-morbidities process}
    \label{fig:multiple_risks}
  \end{center}
  \vspace{-5mm}
\end{figure}

An example of comorbidities is provided by 
\citeauthor{Aalen1980}~\citeyearpar{Aalen1980} 
where two events (onset of menopause, and occurrence of chronic skin disease)
were studied.  Patients could also experience a third event, death.
In this case, we have~$L=2$ binary risk processes each with 
absorbing states and then a final absorbing state of death.

\end{example}

\begin{example}[Competing risks]\label{example:competing_risk}
\normalfont
A patient may experience failure for a multitude of reasons.
Figure~\ref{fig:competing_risks} shows a setting where
failure can be caused by~$L$ risks.  Unlike comorbidities,
a patient may only experience one of the competing risks.

\begin{figure}[!h]
  \begin{center}
    \includegraphics[scale=0.5]{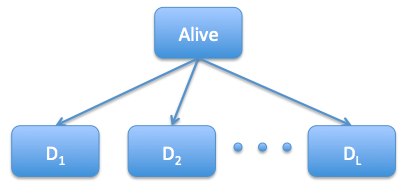}
    \caption{Graph representation of competing risks process}
    \label{fig:competing_risks}
  \end{center}
  \vspace{-5mm}
\end{figure}
\end{example}

\begin{example}[Recurrent events]\label{example:recurrent_events}
\normalfont
Recurrent events are events that occur more than
once per patient.  Examples include recurring
hospital admissions, tumor recurrence, and
repeated heart attacks.
For recurrent events, the state-space is countably
infinite; however, the state-space is structured.
We can assume all patients initial values
are zero~(i.e. $Y(u,0) = 0$) and given~$Y(u,t) = k$
then the transition at the next jump time must
be to~$k+1$. This structure allows us to provide the following
discrete and continuous-time characterizations of the
exchangeable, Markov recurrent event processes -- 
extending Theorems~\ref{thm:discrete_main}
and~\ref{thm:cts_main}.

\begin{corollary}[Discrete-time characterization]
  \label{cor:discrete_recurrent}
  Let~${\bf Y} = (\Y (t), t \in \mathbb{N})$ be a
  discrete-time Markov, exchangeable 
  recurrent event process.  
  Then there exists a probability measure
  $\Sigma$ on $[0,1]^{\mathbb{N}}$ such that~$\Y_{\Sigma}^\star$ is a 
  version of~$\Y$.
\end{corollary}

\begin{proof}[Proof of Corollary~\ref{cor:discrete_recurrent}]
  Consider the recurrent event process~$\Y$ up
  time $\tau<\infty$. Then $\Y^\star = \tau \wedge \Y$
  is a version of $\Y$ on $t \in 0,1,\ldots,\tau$.
  Let~$\eta_{\tau}$ to denote the measure associated
  with~$\Y^{\star}$.  
  For $P \in \P_{\tau}$, let $R_{\tau^\prime, \tau} (A)$ be
  the restriction of this $\P_{\tau^\prime}$.
  Then for $\tau^\prime < \tau$
  \[
    \eta_{\tau} ( \{ P \in \P_{\tau} \given R_{\tau^\prime, \tau} (P) = P^{\star} \})
    = \eta_{\tau^\prime} ( \{ P^{\star} \} )
  \]
  So we have consistency across~$\tau > 0$.  
  We define the measure~$\eta$ on $\P_{\infty}$
  by
  \[
    \eta (\cdot) = \lim_{\tau \uparrow \infty} \eta_{\tau} (\cdot)
  \]
  is the unique measure such that
  $\Y^\star$ is a version of $\Y$.
\end{proof}

\begin{corollary}[Continuous-time characterization]
  \label{cor:cts_recurrent}
  Let~${\bf Y} = (\Y (t), t \in \mathbb{R}^+ )$
  be a continuous-time Markov, exchangeable recurrent event
  process such that $\Y_u (0) = 0$ for all $u \in \mathbb{N}$.
  Let $I$ denote the infinite identity matrix.
  Then there exists a probability measure~$\Sigma$ 
  on~$[0,1]^{\mathbb{N}}$ satisfying
  \[
    \Sigma ( \{ I \} ) = 0
    \text{ and }
    \int_{[0,1]^{\mathbb{N}}} (1-P_{\min}) \Sigma (dp) < \infty 
    \text{ where }
    P_{\min} = \min_{i \in \mathbb{N}} P_{i}
  \]
  and a set of constants~${\bf c} = \{ c_{i, i+1} \given i \in
  \mathbb{N} \}$  such that $\Y_{\Sigma, c}^\star$ is a version
  of~$\Y$.
\end{corollary}

A similar proof can be constructed for the continuous-time setting and
is therefore omitted.

\end{example}

\section{Details on choice of measure}
\label{app:loglinear}

Let~$Z$ be a positive, stationary L{\'e}vy process on~$\mathbb{R}_+$.
As these processes are positive, it is natural to work with the
cumulant function
\[
K(t)  = \log \left( \mathbb{E} \left[ e^{- Z(t) } \right] \right) = 
\log \left( \mathbb{E} \left[ e^{-t X} \right] \right)
\]
for~$t \geq 1$ and $X = Z(1)$ is an infinitely divisible distribution.
The L{\'e}vy-Khintchine characterization for positive, stationary,
L{\'e}vy processes implies
\[
K(t) = - \left[ \gamma t + \int_0^{\infty} (1-e^{-ty}) w(dy) \right]
\]
for some~$\gamma \geq 0$ and measure $w(\cdot)$ on~$\mathbb{R}_+$,
called the L{\'e}vy measure, such that the integral is finite for
all~$t > 0$.

\cite{DempseyMSP} showed that every exchangeable, Markov survival
process can be generated via a L{\'e}vy process construction.
The proof stems from connecting $\gamma$ to the erosion measures
(i.e.,~$c \geq 0$ in Theorem~\eqref{thm:cts_main}) and the L{\'e}vy
measure to the dislocation measures~(i.e., $\Sigma (\cdot)$ in
Theorem~\eqref{thm:cts_main}).
For instance, the harmonic process can be constructed via a L{\'e}vy
process with~$\gamma = 0$ and~$w (dy) = \nu e^{-\rho y} dy / (1-
e^{-y}))$.

Now consider the proportional conditional hazards model as
described by~\cite{Kalbfleisch1978, Hjort1990}, and \cite{Clayton1991}. 
In the proportional conditional hazards model, the hazard
for individual~$i$ is $w_i Z(t)$ for some~$w_i > 0$ typically 
$w_i = \exp (x_i^\prime \beta)$ depending on a set of baseline
covariates~$x_i$.
Then the conditional survival density for particle~$i$ is $\exp \left(
  - w_i \int_0^t Z (t) \right) \left( 1 - e^{-w_i Z (t)} \right)$.
Assume there is a single covariate that is a factor with a finite
number of levels (i.e.,~$x_i \in \{1,\ldots,k\} := [k]$).  Then $w_i =
w_{x_i}$; that is, there are a finite set of weights.
The joint marginal density can be derived in a similar way as before.
Here, however, the non-normalized transition rules are
\begin{align*}
\lambda (R,D) &= \mathbb{E} \left(  e^{- Z (t) \sum_{i \in R} w_{x_i}}
                \prod_{i \in D} \left( 1 - e^{-w_{x_i} Z(t)} \right) 
                \right) \\
              &= \mathbb{E} \left( \prod_{j=1}^{k} \left(
                \exp(-Z(t))^{w_j} \right)^{r_j} \left( 1 - \exp(-Z(t))^{w_j} \right)^{d_j} \right) \\
              &= \int_0^{1} \prod_{j=1}^{k} \left( p^{w_j}
                \right)^{r_j} \left( 1 - p^{w_j} \right)^{d_j} \Sigma (dp)
\end{align*}
where~$r_j = \# \{ i \in R \text{ s.t. } X_i = j \}$,
$d_j = \# \{ i \in D \text{ s.t. } X_i = j \}$, and~$\Sigma(\cdot)$ is
the dislocation measure.
The final equality is due to the connection between the L{\'e}vy
measure and the dislocation measure.
It is clear from above that the proportional conditional hazards
model corresponds to a particular choice of the dislocation measure. 
Namely, the proportional conditional hazards model corresponds to~$p_i \to
p^{w_{x_i}}$.
So on the $[0,1]$-scale, the model is conditionally proportional on
the log-scale. That is,~$\log (p_i ) = w_{x_i} \log(p)$.
Alternative choices exist. For example, the model may be conditionally
proportional on the logistic scale; that is,~$\log \left( p_i / (1-p_i
  ) \right) = w_{x_i} \log \left( p / (1-p) \right)$. We do not pursue
such alternatives in this paper.


\end{document}